\documentclass[a4paper,11pt]{article}

\usepackage{Packages/math}
\usepackage{Packages/page_style}

\title{Ground state energy of the Bose--Hubbard model with large coordination number with a polaron-type quantum de Finetti theorem}

\author{
Shahnaz Farhat\thanks{Department of Mathematics, University of Tübingen, Auf der Morgenstelle 10, 72076 Tübingen, Germany. Email: \texttt{shahnaz.farhat@uni-tuebingen.de}}, 
Denis Périce\thanks{School of Science, Constructor University Bremen, Campus Ring 1, 28759 Bremen, Germany. Email: \texttt{dperice@constructor.university}, corresponding author.}, 
Sören Petrat\thanks{School of Science, Constructor University Bremen, Campus Ring 1, 28759 Bremen, Germany. Email: \texttt{spetrat@constructor.university}}
}

\begin{document}

\maketitle

\begin{abstract}
    We consider the ground state energy of the Bose--Hubbard model on a graph with large and homogeneous coordination number. In the limit of infinite coordination number, we prove convergence of the ground state energy to the minimizer of a mean-field energy functional. This functional is obtained by averaging the hopping term over the large number of connected sites, while the interaction energy is not averaged. Hence, the resulting mean-field description is in the strong coupling regime, and is expected to provide a qualitatively correct picture of the phase diagram of the Bose--Hubbard model for large enough coordination number. For our proof, we develop a new version of a de Finetti-type theorem, which we call the \emph{polaron-type quantum de Finetti theorem}, and which we expect to be a more broadly useful extension of existing quantum de Finetti results. Our theorem covers the case where the Hilbert space is a tensor product of some Hilbert space with a bosonic Fock space. This theorem is applied to the convergence of the ground state energy of the Bose--Hubbard model after reducing it to a polaron-type model.
\end{abstract}


\section{Introduction}

\subsection{Context}

The Bose--Hubbard model \cite{PhysRev.129.959} describes bosons on a lattice with on-site interactions and hopping between nearest neighbours. As a lattice model, it is amenable to rigorous analysis, while being rich in phenomenology. In particular, it is expected to exhibit a quantum phase transition between a Mott insulator and a superfluid state \cite{PhysRevB.40.546}. In the mathematical physics literature, there are only a few rigorous works on this type of phase transition. Without being complete, let us mention the following. The work \cite{Bru_Dorlas} considers the Bose--Hubbard model on a complete graph, which considerably simplifies the problem as the model is then symmetry under permutation of lattice sites. For this model, rigorous proofs on the thermodynamic behaviour are given, and the corresponding phase diagram is analysed. In \cite{ALSSY}, the authors prove Bose--Einstein condensation for the hard-core Bose--Hubbard model at half filling with periodic external potential in the regime of small external potential and low temperature, and the absence of it for large external potential or large temperature. In \cite{MottTransition}, bounds near the critical line of the Mott insulator phase are proven for the hard-core Bose--Hubbard model using rigorous perturbation theory.

In our work, we do not directly attempt to prove the quantum phase transition. Instead, we prove that the ground state energy converges to that of a mean-field model in the limit of infinite coordination number of the lattice. The mean-field model is obtained from averaging the hopping term, and hence the mean-field theory is a one-lattice site description. The corresponding lattice-site product states are known as Gutzwiller product states \cite{Gutzwiller,Rokhsar_Kotliar_1991}. In this mean-field theory one can compute the phase diagram with numerical and some analytical techniques, see Figure~\ref{fig:phase diagram}. It is expected that this mean-field phase diagram is qualitatively correct already for not too large coordination numbers, e.g., already in three dimensions for a cubic lattice with nearest-neighbour interaction \cite{PhysRevB.40.546}. This justifies our interest in the mean-field theory, and the need for a rigorous justification of the validity of the mean-field limit. Note that for the averaging we have to scale down the hopping amplitude with the inverse coordination number, hence our limit could be described as a ``weak-hopping limit''. The interaction is not scaled at all, and thus the mean-field theory is strongly coupled. This is in contrast to weak coupling mean-field limits for bosonic systems, which are not usually expected to exhibit phase transitions; see, e.g., \cite{benedikter_lec,Golse_2016} for reviews and \cite{MFBH} for a result on the Bose--Hubbard model. Note that to our knowledge this is the first proof of convergence of the ground state energy in the limit of large coordination number. Recently, we proved a related result \cite{DMFT_dynamics_2025} for the dynamics. In the physics literature, our approximation is well-known as a simple instance of the celebrated \emph{Dynamical Mean-field Theory} (DMFT), an extremely useful tool to study the dynamics of quantum lattice systems with large coordination number \cite{PhysRevB.40.546,PhysRevB.80.245110}. 

A very versatile and general tool for rigorously justifying mean-field limits is the quantum de Finetti theorem \cite{MR241992,MR397421}. For the justification of mean-field limits in quantum mechanics, it has for example been used in \cite{LEWIN2014570} to rigorously justify Hartree's mean-field theory for continuum bosons in the weak coupling limit. For further references, see the remarks \cite{Lewin_et_al_de_Finetti_2015}, the reviews \cite{Rougerie_notes,Rougerie_review_2021}, and \cite{christandl_2007} for a quantitative version. Many follow-up works used variants of the quantum de Finetti theorem in the derivation of mean-field-like limits in quantum mechanics, e.g., \cite{MNO} for a mixture of Bose gases, \cite{Nam_et_al_2016} for bosons in the Gross--Pitaevskii limit, and \cite{perice_2024} for 2D fermions in strong magnetic fields. The quantum de Finetti theorem states that reduced $k$-body density matrices of a symmetric (i.e., bosonic) state converge in the limit of large particle number to a convex combination of product states. Such a theorem can then be used in lower bounds for ground state energies. For example, for pair interacting particles, the reduced two-body density matrix can be approximated by a convex combination of product states, and a lower bound is obtained by concentration of this convex combination on the minimizers of the corresponding mean-field functional.

The quantum de Finetti theorem as outlined above does not directly apply to our problem of the Bose--Hubbard model with large coordination number. This is because our approximation is for reduced lattice-site density matrices, and the Bose--Hubbard model is not generally symmetric under the exchange of lattice sites. However, we can generate the Bose--Hubbard Hamiltonian through translations of a reduced Hamiltonian of just one lattice site interacting with its neighbouring sites. This reduced Hamiltonian is symmetric in the neighbouring sites. Hence, we prove here a new version of a quantum de Finetti theorem that applies to a tensor product of a Hilbert space (to describe the ``core'' lattice site) with a large symmetric tensor product of another Hilbert space (to describe the neighbouring sites). Such a structure is also encountered when considering a tracer particle coupled to a bath of bosonic particles, e.g., the polaron \cite{polaron_review}, hence we call the corresponding version of such a quantum de Finetti theorem a \emph{polaron-type quantum de Finetti theorem}. Note that we cover the case of infinite dimensional Hilbert spaces here, hence going beyond related finite dimensional versions of such theorems \cite{10.1063/1.2146188,christandl_2007}. We expect this theorem to be interesting in its own right, and regard it as the main technical novelty of the paper. We note that similar results for composite systems have recently been proven in \cite{FOR25}, where they have been applied to Nelson-type polaron models. This was based on the theorem from \cite{FLV88}, and our Theorem~\ref{th:QDF} is in the same spirit as such results. We remark that we explicitly construct the de Finetti measure in the finite dimensional case in Theorem~\ref{th:finite DF}.

\subsection{Model}\label{ssec:model}

We consider a sequence of graphs $(V_z, E_z)_{z\in \mbN}$ with vertex set $V_z$ and edge set $E_z$ with constant coordination number $z$, i.e., 
\begin{align*}
    \forall x \in V_z,\ \abs{V_z^x} \eqcolon z \quad \text{with} \quad V_z^x\coloneq \sett{y \in V \mid \sett{x, y} \in E_z}.
\end{align*}
We will consider the limit where $z \to \infty$. We first notice that 
\begin{align}
    \abs{E_z} = \frac{z}{2} \abs{V_z}, \quad \abs{V_z} > z. \label{eq:z/2}
\end{align}
\begin{remark}\label{rmk:ex graphs}
    Some examples of graphs are:
    \begin{itemize}
        \item The $d$-dimensional square lattice with periodic boundary conditions and length $L \in \mbN^*$, i.e., 
        \begin{align*}
            V_d \coloneq \prth{\mathbb{Z} \slash L \mathbb{Z}}^d,
        \end{align*}
        in the limit $d \to \infty$, with nearest neighbours as edges, so that $z = 2d$.
        \item The cubic lattice $V_3$ with connections inside large balls of radius $r$, so that
        \begin{align*}
            z \eqvl\displaylimits_{r \to \infty} \frac{4}{3}\pi r^3,
        \end{align*}
        with the lattice size satisfying $L \ge r$.
        \item The lattice $V_d$ with neighbours of order $n \in \mbN$ connected, so that
        \begin{align*}
            z \eqvl\displaylimits_{d \to \infty} \frac{(2d)^n}{n!}.
        \end{align*}
    \end{itemize}
\end{remark}

The one-lattice-site Hilbert space is $\ell^2(\mbC)$, and we denote its canonical Hilbert basis by $(\ket{n})_{n \in \mbN}$. We define the standard annihilation and creation operators $a, a^*$ satisfying the CCR by
\begin{align*}
    &a \ket{0} = 0, \quad \forall n \in \mbN^*,\ a \ket{n} \coloneq \sqrt{n}\,\ket{n-1}, \\
    &\forall n \in \mbN,\ a^* \ket{n} \coloneq \sqrt{n+1}\,\ket{n+1}.
\end{align*}
The particle number operator is
\begin{align*}
    \mathcal{N} \coloneq a^* a,
\end{align*}
and the Fock space is
\begin{equation}
    \mcF_{V_z} \coloneq \ell^2(\mbC)^{\otimes \abs{V_z}} \cong \mcF_+\prth{L^2(V_z, \mbC)} \coloneq \bigoplus_{n \in \mbN} L^2(V_z, \mbC)^{\otimes_+ n}. \label{eq:eqv fock spaces}
\end{equation}
In these notations, $\otimes_+$ denotes the symmetric tensor product and $\mcF_+\prth{L^2(V_z, \mbC)}$ the bosonic Fock space constructed over the one-particle Hilbert space $L^2(V_z, \mbC)$. Equation~\eqref{eq:eqv fock spaces} provides two equivalent representations of the same Fock space. The right-hand side is the standard ``particle representation'', while we call the left-hand side the ``lattice-site representation''. The latter will be the convenient setting for this paper, as we aim to control correlations between lattice sites rather than between particles.

We denote by $\mcL$ the set of linear operators, by $\mcL_+$ the set of positive linear operators, by $\mcL^p$ the $p$-Schatten class for $p \in \sbra{1, \infty}$, and by $\mathcal{K}$ the set of compact operators. Given an order $\le$ on $V_z$, we denote
\begin{align}
    \forall A \in \mcL\prth{\ell^2(\mbC)},\ x \in V_z,\ A_x \coloneq \mathds{1}_{\ell^2}^{\otimes \abs{\sett{y \in V \mid y < x}}} \otimes A \otimes \mathds{1}_{\ell^2}^{\otimes \abs{\sett{y \in V \mid x < y}}} \in \mcL \prth{\mcF_z}. \label{eq:A_x}
\end{align}
For a hopping amplitude $J \in \mbR$, a chemical potential $\mu \in \mbR$, and a coupling constant $U \in \mbR$, we define the Bose--Hubbard Hamiltonian
\begin{align}
    H_{V_z} \coloneq - \frac{J}{z} \sum_{\sett{x, y} \in E_z} \prth{a_x^* a_y + a_y^* a_x}
    + (J - \mu) \sum_{x \in V_z} \mathcal{N}_x
    + \frac{U}{2} \sum_{x \in V_z} \mathcal{N}_x (\mathcal{N}_x - 1). \label{eq:HBH}
\end{align}
In view of \eqref{eq:z/2}, we note that the kinetic energy is given by
\begin{align*}
    d\Gamma(-\Delta)
        \coloneq&\ \sum_{\sett{x, y} \in E_z} (a_x^* - a_y^*)(a_x - a_y)
        = \sum_{\sett{x, y} \in E_z} \prth{\mcN_x + \mcN_y}
        - \sum_{\sett{x, y} \in E_z} \prth{a_x^* a_y + a_y^* a_x} \\
        =&\ z \sum_{x \in V_z} \mcN_x
        - \sum_{\sett{x, y} \in E_z} \prth{a_x^* a_y + a_y^* a_x},
\end{align*}
with scaled amplitude $J/z$, so that the energy contribution is of the same order as that of the one-lattice-site terms when $z\to\infty$. We emphasize that the mean-field limit considered in this work amounts to averaging over the $y$ variable in \eqref{eq:HBH}, i.e. over the neighbouring shell $V_z^x$ around $x \in V_z$. This is the main difference from known results on mean-field limits, which average over particle interactions \cite{benedikter_lec,Golse_2016,MFBH}. This is also the reason why the mean-field scaling $1/z$ appears in front of the kinetic energy in \eqref{eq:HBH}, as opposed to the usual $1/N$ rescaling of the interaction energy.

Note that, using \eqref{eq:z/2}, the Bose--Hubbard Hamiltonian \eqref{eq:HBH} can be rewritten for any $\alpha \in \C$ as
\begin{align*}
H_{V_z} &=  \sum_{x \in V_z} \prth{-J \prth{\alpha a_x^*  + \overline{\alpha} a_x - \abs{\alpha}^2} + (J - \mu) \mathcal{N}_x + \frac{U}{2}\mathcal{N}_x (\mathcal{N}_x - 1)} \\
&\quad - \frac{J}{z} \sum_{ \sett{x, y} \in E_z} \prth{(a_x^*-\overline{\alpha}) (a_y-\alpha) + (a^*_y-\overline{\alpha})(a_x-\alpha)}.
\end{align*}
Then the mean-field theory is obtained by neglecting all terms quadratic in $a-\alpha$, and choosing $\alpha$ as the averaged annihilation operator. Hence, we introduce the nonlinear mean-field operator
\begin{equation}
h_\varphi \coloneq - J \big( \alpha_\varphi \ad + \overline{\alpha_\varphi} a - \abs{\alpha_\varphi}^2 \big) + (J - \mu) \mathcal{N} + \frac{U}{2} \mathcal{N}(\mathcal{N}-1), \label{eq:h mf}
\end{equation}
with the order parameter
\begin{align}
    \alpha_\varphi \coloneqq \scp{\varphi}{a \varphi} \label{eq:alpha}
\end{align}
where $\varphi \in \ell^2(\mbC)$ is a one-lattice-site state. The corresponding energy functional is
\begin{equation}\label{mf_energy_func}
\bk{\varphi}{h_\varphi \varphi} = - J |\alpha_\varphi|^2 + (J - \mu) \bk{\varphi}{\mathcal{N}\varphi} + \frac{U}{2} \bk{\varphi}{\mathcal{N}(\mathcal{N}-1)\varphi}.
\end{equation}

\subsection{Main results}

Our main theorem is the convergence of the ground state energy of the Bose--Hubbard model \eqref{eq:HBH} to the ground state energy of the corresponding mean-field energy functional \eqref{mf_energy_func}.

\begin{theorem}\label{th:energy cvg}
    If $U > 0$ and $J \ge 0$, then
    \begin{align*}
        \inf_{\substack{\psi \in \mcF_{V_z} \\ \norm{\psi}_{\mcF_{V_z}}=1}} \frac{\bk{\psi}{H_{V_z} \psi}}{\abs{V}} 
            \cvgto_{z\to\infty} \inf_{\substack{\varphi \in \ell^2(\mbC) \\ \norm{\varphi}_{\ell^2} = 1}} \bk{\varphi}{h_\varphi \varphi}.
    \end{align*}
\end{theorem}
The proof is given in Section~\ref{ch_energy_bound}.

\begin{remark}\label{rmk:diagram discussion 1}
Let us gather a few simple facts about the ground state energy of the Bose--Hubbard model \eqref{eq:HBH}. Note that we do not fix the particle number (or density) as a constraint. Then the vacuum state is a possible trial state, hence the ground state energy is always non-positive. Furthermore, for example, for $U<0$, we can use a sequence of trial states with increasing density to see that in this case there is no ground state, i.e., the ground state energy is $-\infty$. Then there are only a few non-trivial regimes of parameters $J,\mu,U \in \R$, i.e., regimes where the ground state energy is neither $0$ nor $-\infty$. The most physically relevant regime corresponds to $J\geq 0$, $\mu >0$, and $U>0$. The physically less relevant regimes, with non-trivial ground state, occur for negative hopping amplitude, i.e., when $J<0$, together with $U>0$ and $\mu > J$.
\end{remark}

\begin{remark}\label{rmk:diagram discussion 2}
    In the regime where $J\geq 0$, $\mu >0$, and $U>0$, one expects a superfluid phase for $\frac{J}{U} \gg 1$. For small values of $\frac{J}{U}$, a Mott insulating phase can occur. Mean-field theory predicts a phase diagram as in Figure~\ref{fig:phase diagram}, and this prediction is believed to be qualitatively correct already for relatively small coordination numbers, e.g., for a simple cubic lattice (three dimensions) with nearest neighbour interaction.
\end{remark}

\begin{figure}[H]
  \begin{minipage}[c]{0.55\textwidth}
    \includegraphics[width=0.99\linewidth]{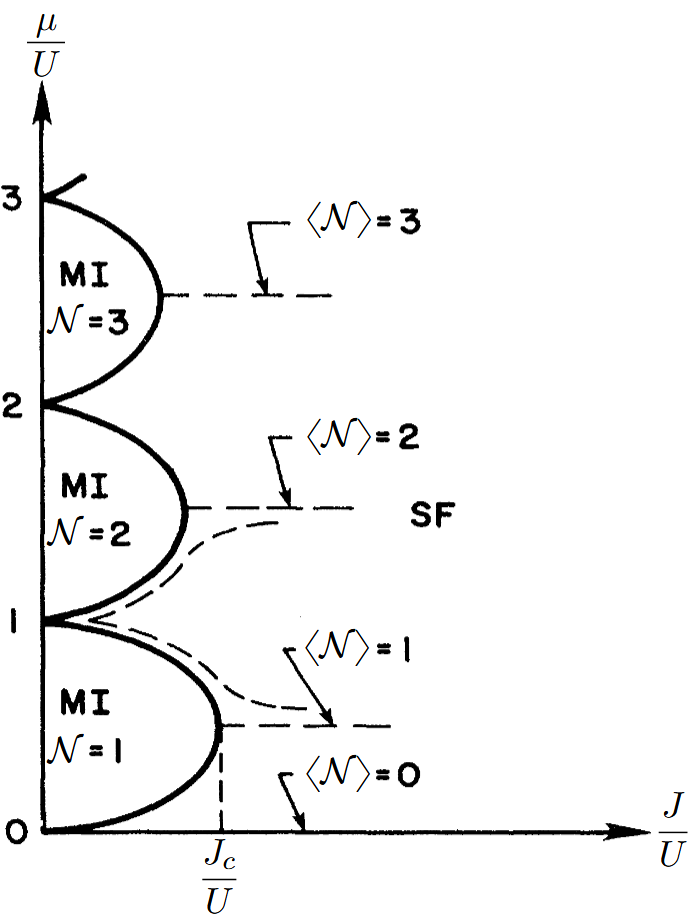}
  \end{minipage}\hfill
  \begin{minipage}[c]{0.44\textwidth}

    \caption{Mott insulator (MI) / Superfluid (SF) phase diagram \cite{PhysRevB.40.546}. The diagram is obtained by minimizing the mean-field energy functional \eqref{mf_energy_func}. Provided a ground state $\varphi(J, \mu, U)$, the Mott insulator phase is defined as region where $\alpha_\varphi(J, \mu, U) = 0$.}
    \label{fig:phase diagram}
  \end{minipage}
\end{figure}
The upper bound in Theorem~\ref{th:energy cvg} follows directly from a simple trial state argument using a lattice-site product state, see Proposition~\ref{prop:UB}. The lower bound is non-trivial since the Bose--Hubbard Hamiltonian \eqref{eq:HBH} is in general not symmetric under the exchange of lattice sites. However, in Section~\ref{ssec:Hz} we will show that in the proof of a lower bound it can be reduced to the local Hamiltonian
\begin{align}\label{eq_H_loc_z}
    H_{1, z} \coloneq \sum_{i=1}^z \prth{- J \prth{ a_0^* a_i + a_i^* a_0} + (J-\mu) \prth{\mcN_0 + \mcN_i} + \frac{U}{2} \prth{\mcN_0(\mcN_0-1) + \mcN_i(\mcN_i-1)}},
\end{align}
which is symmetric with respect to exchanging the variables $i\in \intint{1,z}$. Our strategy is to use a quantum de Finetti theorem suitable for such Hamiltonians in the proof of the lower bound. 

The system described by the Hamiltonian \eqref{eq_H_loc_z} behaves like a system consisting of one particle of a first type (corresponding to the index $i=0$) and $z$ bosonic particles of a second type (indices $i\in \intint{1,z})$. Hence, we introduce the following formalism for bosonic systems with multiple species.

\begin{definition}[Multiple-species bosonic states and reduced density matrices]\label{def:mult bosons}
    Let $\mcH_1, \mcH_2$ be two separable Hilbert spaces. Given $N_1, N_2 \in \mbN$, an operator
    \begin{align*}
        \gamma_{N_1, N_2} \in \mcL^1_+\prth{\mcH_1^{\otimes_+ N_1} \otimes \mcH_2^{\otimes_+ N_2}} \text{ satisfying } \Tr{\gamma_{N_1, N_2}} = 1
    \end{align*}
    is called a $(N_1, N_2)$-bosonic state. Let $k_1 \in \intint{0, N_1}$ and $k_2\in \intint{0, N_2}$. We define
    \begin{align*}
        \gamma_{N_1, N_2}^{(k_1, k_2)} \coloneq \PTr{\mcH_1^{\otimes (N_1-k_1)} \otimes \mcH_2^{\otimes (N_2-k_2)}}{\gamma_{N_1, N_2}}
    \end{align*}
    as the $(k_1, k_2)$-reduced density matrix, where for $i\in \sett{1, 2}$ we trace out $N_i - k_i$ variables from the symmetric variables in $\mcH_i$.
\end{definition}

Note that this definition can be generalized to more than two species of particles. Next, we generalize to Fock spaces.

\begin{definition}[Multiple-species infinite bosonic states]\label{def:inf bosons states}
    We extend Definition~\ref{def:mult bosons} to the case where $N_2 = \infty$. An $(N_1, \infty)$-bosonic state is a sequence
    \begin{align*}
        \gamma_{N_1, \infty} \coloneq \prth{\gamma_{N_1, N_2}}_{N_2 \in\mbN} \in \mcL\prth{\mcH_1^{\otimes_+ N_1} \otimes \mcF_+\prth{\mcH_2}}
    \end{align*}
    where $\gamma_{N_1, N_2} \in \mcL\prth{\mcH_1^{\otimes_+ N_1} \otimes \mcH_2^{\otimes_+ N_2}}$ are $(N_1, N_2)$-bosonic states satisfying the consistency condition
    \begin{align*}
        \forall N_2 \in \mbN, \ \gamma_{N_1, N_2+1}^{(N_1, N_2)} = \gamma_{N_1, N_2}.
    \end{align*}
    Then we can define
    \begin{align*}
        \forall k_1 \in \intint{0:N_1},\ k_2\in \mbN,\ \gamma_{N_1, \infty}^{(k_1, k_2)} \coloneq \gamma_{N_1, N_2}^{(k_1, k_2)} \text{ for any } N_2 \ge k_2.
    \end{align*}
\end{definition}

This definition can be extended to the case $N_1 = \infty$ and to more than two species of bosons. 

For such states, we have the following quantum de Finetti theorem. Let us denote by $S_{\mcH_s}$ the sphere in $\mcH_s$, and by $\mcM$ the set of Radon measures.

\begin{theorem}[Polaron quantum De Finetti in the limit $N\to\infty$] \label{th:QDF}
    Let $\gamma \in \mcL\prth{\mcH_0 \otimes  \mcF_+\prth{\mcH_s}}$ be a $(1,\infty)$-bosonic state. We assume that there exists a sequence $(P_m)_{m\in\mbN}$ of projectors on $\mcH_s$, respectively of rank $m+1$, such that
    \begin{align*}
        \Tr{\gamma^{(0, 1)}P_m} \cvgto_{m\to\infty} 1.
    \end{align*}
    Then there exist
    \begin{itemize}
        \item a probability measure $\mbP \in \mcM\prth{S_{\mcH_s},\mbR_+}$,
        \item a function $\zeta \in L^1\prth{S_{\mcH_s}, \mcL^1_+(\mcH_0)}$ satisfying, $\mbP$-a.e., $\Tr{\zeta} = 1$,
    \end{itemize}
    such that
    \begin{align}
        \forall k\in \mbN^*,\ \gamma^{(1, k)} = \intr_{S_{\mcH_s}} \zeta(u) \otimes p_u^{\otimes k}  \ d\mbP(u).
    \end{align}
\end{theorem}

The theorem is proven in Section~\ref{sec_de_Finetti}. 

The remainder of the paper proceeds as follows. In Section~\ref{ch_energy_bound}, we provide a proof of Theorem~\ref{th:energy cvg}. First, in Section~\ref{ch_upper_bound}, we prove the upper bound. In Section~\ref{ssec:Hz}, we discuss the reduction of the Bose--Hubbard Hamiltonian \eqref{eq:HBH} to the reduced Hamiltonian \eqref{eq_H_loc_z}. Additionally, we prove bounds on certain moments of number operators that are used to verify the assumptions of Theorem~\ref{th:QDF}. With that, we conclude the proof of Theorem~\ref{th:energy cvg} in Section~\ref{ch_proof_of_GS_Convergence}. Section~\ref{sec_de_Finetti} is devoted to proving Theorem~\ref{th:QDF}. First, in Section~\ref{sec_Finetti_finite}, we deal with the finite dimensional case, and in Section~\ref{sec_Fock_space_loc}, we then use Fock space localization methods to prove the infinite dimensional case.

\section{Ground state energy convergence}\label{ch_energy_bound}

\subsection{Upper energy bound}\label{ch_upper_bound}

The upper bound is trivial in the sense that it is sufficient to take a factorized state as a trial state. Indeed, the Bose--Hubbard energy per lattice site for a lattice-site-factorized state is equal to the mean-field energy.

\begin{proposition}[Upper energy bound]\label{prop:UB}
    Let $\varphi \in \ell^2(\mbC)$ such that $\norm{\varphi}_{\ell^2} = 1$, then
    \begin{align}
        \frac{\bk{\varphi^{\otimes \abs{V_z}}}{H_{V_z} \varphi^{\otimes \abs{V_z}}}}{\abs{V_z}}
            = \bk{\varphi}{h_\varphi \varphi} \label{eq:hmf expression}
    \end{align}
    and hence
    \begin{align}
        \inf_{\substack{\psi \in \mcF_{V_z} \\ \norm{\psi}_{\mcF_{V_z}}=1}} \frac{\bk{\psi}{H_{V_z} \psi}}{\abs{V}}
            \le \inf_{\substack{\varphi \in \ell^2(\mbC) \\ \norm{\varphi}_{\ell^2} = 1}} \bk{\varphi}{h_\varphi \varphi}. \label{eq:UB}
    \end{align}
\end{proposition}

\begin{proof}
    Using \eqref{eq:HBH} and then \eqref{eq:z/2}, \eqref{eq:alpha}, and \eqref{mf_energy_func}, we obtain
    \begin{align*}
        \frac{\bk{\varphi^{\otimes \abs{V_z}}}{H_{V_z} \varphi^{\otimes \abs{V_z}}}}{\abs{V_z}}
            = -\frac{2J \abs{E}}{z \abs{V_z}} \bk{\varphi}{a^* \varphi}\bk{\varphi}{a \varphi}
                + (J-\mu) \bk{\varphi}{\mathcal{N}\varphi} + \frac{U}{2}\bk{\varphi}{\mcN(\mcN - 1) \varphi}
            = \bk{\varphi}{h_\varphi \varphi}.
    \end{align*}
    Equation~\eqref{eq:UB} follows by minimizing over $\varphi$ and noticing that $\varphi^{\otimes \abs{V_z}} \in \mcF_{V_z}$ with $\norm{\varphi^{\otimes \abs{V_z}}}_{\mcF_{V_z}}=1$.
\end{proof}

\subsection{Translation invariance}\label{ssec:Hz}

Noting that \eqref{eq:HBH} is ``translation invariant'', we reduce the model to the local Hamiltonian
\begin{align}
H_{1, z} \coloneq \sum_{i=1}^z \prth{- J \prth{ a_0^* a_i + a_i^* a_0} + (J-\mu) \prth{\mcN_0 + \mcN_i} + \frac{U}{2} \prth{\mcN_0(\mcN_0-1) + \mcN_i(\mcN_i-1)}} \label{eq:Hz}
\end{align}
acting on $\ell^2(\mbC)^{\otimes (z+1)}$. As the graph is not embedded in a vector space, the translations we refer to are formally maps of the form $T: V_z \to V_z$ satisfying
\begin{align*}
\forall x \in V_z,\ T(V_z^x) = V_z^{T(x)}.
\end{align*}
For instance, in the examples we provide in Remark~\ref{rmk:ex graphs}, physical translations by a lattice vector are translations in this sense.

The Hamiltonian \eqref{eq:Hz} is symmetric with respect to the index $i$, i.e., with respect to its last $z$ variables. Hence, we consider the Hilbert space
\begin{align*}
\mcF_{1, z} \coloneq \ell^2(\mbC)\otimes \ell^2(\mbC)^{\otimes_+ z}.
\end{align*}
This is a Fock space sector with one distinguished lattice site, which we call the core, and $z$ indistinguishable lattice sites, called the (neighbouring) shell.

Let $\gamma \in \mcL^1\prth{\mcF_{V_z}}$ and $W \subseteq V_z$. We define the partial trace of $\gamma$ over $V_z \backslash W$ as the operator $\PTr{V_z \backslash W}{\gamma}$ on $\mcL^1\prth{\ell^2(\mbC)^{\otimes \abs{W}}}$ such that
\begin{align}
    \forall K \in \mathcal{K}\prth{\ell^2(\mbC)^{\otimes \abs{W}}},\ \Tr{\PTr{V_z \backslash W}{\gamma} K} \coloneq \Tr{\gamma K_W}, \label{eq:ptr}
\end{align}
with $K_W$ being the generalization of \eqref{eq:A_x} acting on the coordinates in $W$ in the order provided by the ordered set $(V_z, \le)$.

In the following, we would like to specify the order of action on the coordinates of $V_z$ such that the first coordinate corresponds to core variables. We proceed by considering a list $w \coloneq \prth{w_{1:\abs{W}}}$ such that $W = \sett{w_i,\ i \in \intint{1:\abs{W}}}$ and generalize \eqref{eq:ptr} through
\begin{align}
    \forall K \in \mathcal{K}\prth{\ell^2(\mbC)^{\otimes \abs{W}}},\ \Tr{\PTr{V_z \backslash w}{\gamma} K} \coloneq \Tr{\gamma K_w}, \label{eq:ptr-w}
\end{align}
where $K_w$ acts on the coordinates in $W$ in the order provided by the list $w$. When $w_1 \le w_2 \le \dots \le w_{\abs{W}}$, we recover
\begin{align*}
    \PTr{V_z \backslash w}{\gamma} = \PTr{V_z \backslash W}{\gamma}.
\end{align*}
We then introduce the $(1, z)$-reduced density matrix
\begin{align}
    \gamma^{(1, z)} \coloneq \frac{1}{\abs{V_z}} \sum_{x \in V_z} \PTr{V_z \backslash (x, V_z^x)}{\gamma}. \label{eq:gamma1z}
\end{align}
Note that specifying an order on the shell $V_z^x$ is irrelevant due to the symmetry of $H_{1, z}$. This is why, in \eqref{eq:gamma1z}, we only specify that the core variable $x$ corresponds to the first coordinate in $\gamma^{(1, z)}$. We then have the following result.

\begin{proposition}[Energy reduction via translation invariance]\label{prop:TIR}
For $H_{V_z}$ from \eqref{eq:HBH} and $H_{1,z}$ from \eqref{eq:Hz} we have
    \begin{align*}
        \inf_{\substack{\psi_{1, z} \in \mcF_{1, z} \\ \norm{\psi_{1, z}}_{\mcF_{1, z}}=1}} \frac{\bk{\psi_{1, z}}{H_{1, z} \psi_{1, z}}}{2z}
            = \inf_{\substack{\psi_{1, z} \in \ell^2(\mbC)^{\otimes(z+1)} \\ \norm{\psi_{1, z}}_{\ell^2}=1}} \frac{\bk{\psi_{1, z}}{H_{1, z} \psi_{1, z}}}{2z}
            \le \inf_{\substack{\psi \in \mcF_{V_z} \\ \norm{\psi}_{\mcF_{V_z}}=1}} \frac{\bk{\psi}{H_{V_z} \psi}}{\abs{V_z}}.
    \end{align*}
\end{proposition}
\begin{proof} We find
    \begin{align*}
        \sum_{x \in V} (H_{1, z})_{x, V_z^x}
            =&\ \sum_{x \in V} \sum_{y \in V_z^x} \prth{- J \prth{a_x^* a_y + a_y^* a_x} + (J-\mu) \prth{\mcN_x + \mcN_y} + \frac{U}{2} \prth{\mcN_x(\mcN_x-1) + \mcN_y(\mcN_y-1)}} \\
            =&\ 2z H_{V_z}.
    \end{align*}
    Let $\psi \in \mcF_{V_z}$ with $\norm{\psi}_{\mcF_{V_z}} = 1$. We then have
    \begin{align*}
        \frac{\bk{\psi}{H_{V_z} \psi}}{\abs{V_z}}
            =&\ \frac{1}{2z \abs{V_z}} \sum_{x \in V} \bk{\psi}{(H_{1, z})_{x, V_z^x} \psi}
            = \frac{1}{2z \abs{V_z}} \sum_{x \in V} \Tr{\PTr{V \backslash (x, V_z^x)}{\ket{\psi}\bra{\psi}} H_{1, z}} \\
            =&\ \frac{\Tr{\ket{\psi}\bra{\psi}^{(1, z)} H_{1, z}}}{2z}
            \ge \inf_{\substack{\psi_{1, z} \in \ell^2(\mbC)^{\otimes(z+1)} \\ \norm{\psi_{1, z}}_{\ell^2}=1}} \frac{\bk{\psi_{1, z}}{H_{1, z} \psi_{1, z}}}{2z}.
    \end{align*}
    The first equality holds due to the symmetry of $H_{1, z}$ with respect to its last $z$ variables.
\end{proof}

Next, we estimate a useful combination of number operators. We introduce the following notation.
\begin{definition}[Moments operators]\label{def:moments}
We define
    \begin{align*}
    \forall A \in \mcL^1\prth{\ell^2(\mbC)},\ A_s \coloneq \frac{1}{z}\sum_{i=1}^z A_i,\quad
    \forall \beta\in \mbR_+,\ \mcM_\beta
        \coloneq \sum_{i=1}^z \frac{\mcN_0^\beta + \mcN_i^\beta}{z}
        = \mcN_0^\beta + (\mcN^\beta)_s.
    \end{align*}
    This defines operators on $\ell^2(\mbC) \otimes \ell^2(\mbC)^{\otimes_+ z}$, and we also extend them, sector-wise, to the Fock space
    \begin{align*}
        \ell^2(\mbC) \otimes \mcF_+\prth{\ell^2(\mbC)} \simeq \bigoplus_{z\in \mbN} \ell^2(\mbC) \otimes \ell^2(\mbC)^{\otimes_+ z}.
    \end{align*}
\end{definition}

The Hamiltonian \eqref{eq:Hz} can then be written as
\begin{align}
    \frac{H_{1, z}}{z} = -J\prth{a_0^* a_s + a_s^* a_0} + \prth{J-\mu -\frac{U}{2}}\mcM_1 + \frac{U}{2}\mcM_2. \label{eq:Hz/z}
\end{align}
We note the following bound.

\begin{proposition}
    Let $0 \le \beta_1 \le \beta_2$. Then,
    \begin{align}
        \mcM_{\beta_1} \le 2^{1-\frac{\beta_1}{\beta_2}} \mcM_{\beta_2}^{\frac{\beta_1}{\beta_2}}. \label{eq:M holder}
    \end{align}
\end{proposition}

\begin{proof}
    Let
    \begin{align*}
        \forall i \in \intint{1:2z},\quad A_i \coloneq \syst{\mcN_0^{\beta_1} \text{ if $i$ is odd,} \\ \mcN_{\frac{i}{2}}^{\beta_1} \text{ if $i$ is even.}}
    \end{align*}
    By Hölder's inequality,
    \begin{align*}
        \mcM_{\beta_1} 
            = \frac{1}{z}\sum_{i=1}^z A_i 
            \le \frac{1}{z} (2z)^{1-\frac{\beta_1}{\beta_2}} \, \prth{\sum_{i=1}^{2z} A_i^{\frac{\beta_2}{\beta_1}}}^{\frac{\beta_1}{\beta_2}}
            = 2^{1-\frac{\beta_1}{\beta_2}} \, \prth{\sum_{i=1}^z \frac{\mcN_0^{\beta_2} + \mcN_i^{\beta_2}}{z}}^{\frac{\beta_1}{\beta_2}}
            = 2^{1-\frac{\beta_1}{\beta_2}} \, \mcM_{\beta_2}^{\frac{\beta_1}{\beta_2}}.
    \end{align*}
\end{proof}

We can then prove an energy estimate for the moment $\mcM_2$. For $x\in \mbR$, we set $x_\pm \coloneq \max(0, \pm x)$.
\begin{proposition}[$\mcM_2$ energy estimate]\label{prop:M2 esti}
    If $U > 0$, then
    \begin{align*}
        &2J_- + \mu + \frac{U}{2} \le 0 \implies \mcM_2 \le \frac{2}{U}\,\frac{H_{1, z}}{z},\\
        &2J_- + \mu + \frac{U}{2} > 0 \implies \mcM_2 \le \frac{4}{U}\,\frac{H_{1, z}}{z} + 2\prth{\frac{4J_- + 2\mu}{U}+1}^2.
    \end{align*}
\end{proposition}
\begin{proof}
    Noticing that
    \begin{align}
        0 \le \frac{1}{z}\sum_{i=1}^z (a_0^* \pm a_i^*)(a_0 \pm a_i)
            = \mcM_1 \pm \prth{a_0^* a_s + a_s^* a_0} 
            \implies&\ - \mcM_1 \le a_0^* a_s + a_s^* a_0 \le \mcM_1,  \nonumber\\
            \implies&\ \abs{a_0^* a_s + a_s^* a_0} \le \mcM_1,  \label{eq:laplacian CS}
    \end{align}
    and using \eqref{eq:M holder} with $\beta_1 \coloneq 1, \beta_2 \coloneq 2$ together with \eqref{eq:Hz/z}, we estimate, for $\epsilon > 0$,
    \begin{align*}
        \frac{U}{2}\mcM_2  
            =&\ \frac{H_{1, z}}{z} + J\prth{a_0^* a_s + a_s^* a_0}
                - \prth{J - \mu - \frac{U}{2}}\mcM_1
            \le \frac{H_{1, z}}{z}
                + \prth{2J_- + \mu + \frac{U}{2}}\mcM_1 \\
            \le&\ \frac{H_{1, z}}{z}
                + \prth{2J_- + \mu + \frac{U}{2}}_+ \sqrt{2 \mcM_2}
            \le \frac{H_{1, z}}{z}
                + \prth{2J_- + \mu + \frac{U}{2}}_+
                    \prth{\frac{1}{2\epsilon} + \epsilon\mcM_2}.
    \end{align*}
    If $2J_- + \mu + \frac{U}{2} \le 0$, one has
    \begin{align*}
        \mcM_2 \le \frac{2}{U}\frac{H_{1, z}}{z}.
    \end{align*}
    Otherwise, choosing $\epsilon \coloneq \frac{U}{4\prth{2J_- + \mu + \frac{U}{2}}}$, we find
    \begin{align*}
        \mcM_2 \le \frac{4}{U}\,\frac{H_{1, z}}{z}
            + 2\prth{\frac{4J_- + 2\mu}{U}+1}^2.
    \end{align*}
\end{proof}

\begin{remark}[Domain of $H_{1, z}$]\label{rk:N sobolev}
    If $\mcH$ is a separable Hilbert space and $A \in \mcL\prth{\mcH}$ satisfies $A \ge \mathds{1}_{\mcH}$, then the Sobolev space
    \begin{align*}
        \prth{\mcL^{1, A}\prth{\mcH} \coloneq
        \sett{\gamma \in \mcL^1(\mcH) \ \big| \ \sqrt{A} \, \gamma \, \sqrt{A} \in \mcL^1\prth{\mcH}},\
        \norm{\bullet}_{\mcL^{1, A}}\coloneq \norm{\sqrt{A} \, \bullet \, \sqrt{A}}_{\mcL^1}}
    \end{align*}
    is a Banach space whose topological pre-dual is
    \begin{align*}
        \prth{
            \sqrt{A} \, \mcK(\mcH) \, \sqrt{A},\
            \norm{A^{-\frac{1}{2}}
                \bullet
                A^{-\frac{1}{2}}}_{\mcL^\infty}
        }.
    \end{align*}
    Using \eqref{eq:laplacian CS} and \eqref{eq:M holder} for the Hamiltonian \eqref{eq:Hz/z}, we obtain
    \begin{align*}
        \frac{H_{1, z}}{z}
            \le \prth{\abs{J} + J - \mu - \frac{U}{2}}\mcM_1 + \frac{U}{2}\mcM_2
            \le \prth{2J_+ - \mu - \frac{U}{2}}\sqrt{2\mcM_2}
                + \frac{U}{2}\mcM_2.
    \end{align*}
    This shows that $H_{1, z}$ has domain $\mcL^{1, \mcM_2 + \mathds{1}}\prth{\ell^2(\mbC) \otimes \ell^2(\mbC)^{\otimes_+ z}}$. We denote
    \begin{align*}
        \forall \beta \in \mbR_+,\ \mcL^{1, \beta} \coloneq \mcL^{1, \mcM_\beta + \mathds{1}}.
    \end{align*}
\end{remark}

An important step toward the lower bound is the weak$^*$ lower semicontinuity of $H_{1, z}$ on its domain.

\begin{proposition}[Weak$^*$ lower semicontinuity of $H_{1, z}$]\label{prop:w*lsc Hz}
    Let $k \in \mbN^*$ and
    \begin{align*}
        (\gamma_i)_{i\in \mbN} \subseteq \mcL^{1, 2}_+\prth{\ell^2(\mbC) \otimes \ell^2(\mbC)^{\otimes_+ k}}
    \end{align*}
    such that
    \begin{align*}
        \gamma_i \wslim\displaylimits_{i\to\infty} \gamma \in \mcL^{1, 2}\prth{\ell^2(\mbC) \otimes \ell^2(\mbC)^{\otimes_+ k}}.
    \end{align*}
    If $U \ge 0$, then for all $\beta \in [0, 2[$, we have strong convergence in $\mcL^{1, \beta}\prth{\ell^2(\mbC) \otimes \ell^2(\mbC)^{\otimes_+ k}}$ and
    \begin{align*}
        \frac{\Tr{\gamma H_{1, k}}}{k} \le \liminf_{i\to\infty} \frac{\Tr{\gamma_i H_{1, k}}}{k}.
    \end{align*}
\end{proposition}

\begin{proof}
    Let $\beta \in [0, 2[$. Using \eqref{eq:M holder},
    \begin{align*}
        \prth{\mcM_2 + \mathds{1}}^{-\frac{1}{2}} \mcM_\beta \prth{\mcM_2 + \mathds{1}}^{-\frac{1}{2}} 
            \le&\ 2^{1-\frac{\beta}{2}} \mcM_2^{\frac{\beta}{2}}\prth{\mcM_2 + \mathds{1}}^{-1} \\
            \le&\ 2^{1-\frac{\beta}{2}} \prth{\mcM_2 + \mathds{1}}^{-\frac{2 - \beta}{2}} \in \mcK\prth{\ell^2(\mbC) \otimes \ell^2(\mbC)^{\otimes_+ k}},
    \end{align*}
    hence $\mcM_\beta \in \sqrt{\mcM_2 + \mathds{1}} \, \mcK\prth{\ell^2(\mbC) \otimes \ell^2(\mbC)^{\otimes_+ k}} \, \sqrt{\mcM_2 + \mathds{1}}$, and we can pass to the limit in
    \begin{align}
        \Tr{\gamma_i \mcM_\beta} \cvgto_{i\to\infty}& \Tr{\gamma \mcM_\beta}. \label{eq:beta mass lim}
    \end{align}
    The case $\beta = 0$ yields $\Tr{\gamma_i} \cvgto\limits_{i\to\infty} \Tr{\gamma}$, and since $\gamma \ge 0$ (non-negativity passes to the limit), we also have
    \begin{align*}
        \norm{\gamma_i}_{\mcL^{1, \beta}} 
            = \Tr{\gamma_i \prth{\mcM_\beta + \mathds{1}}}
            \cvgto_{i\to\infty} \Tr{\gamma \prth{\mcM_\beta + \mathds{1}}}
            = \norm{\gamma}_{\mcL^{1, \beta}},
    \end{align*}
    hence the strong convergence (see \cite{simon2005trace} for a reference). Similarly,
    \begin{align*}
        \prth{\mcM_2 + \mathds{1}}^{-\frac{1}{2}} \abs{a_0^*a_s + a_s^* a_0} \prth{\mcM_2 + \mathds{1}}^{-\frac{1}{2}} 
            \le \prth{\mcM_2 + \mathds{1}}^{-\frac{1}{2}} \mcM_1 \prth{\mcM_2 + \mathds{1}}^{-\frac{1}{2}} \in \mcK\prth{\ell^2(\mbC) \otimes \ell^2(\mbC)^{\otimes_+ k}},
    \end{align*}
    therefore
    \begin{align}
        \Tr{\gamma_i \prth{a_0^*a_s + a_s^* a_0}} \cvgto_{i\to\infty}& \Tr{\gamma \prth{a_0^*a_s + a_s^* a_0}}. \label{eq:kim part}
    \end{align}
    Moreover,
    \begin{align}
        \Tr{\gamma \mcM_2} 
            =&\ \Tr{\prth{\mcM_2 + \mathds{1}}^{\frac{1}{2}} \gamma \prth{\mcM_2 + \mathds{1}}^{\frac{1}{2}}} - \Tr{\gamma} 
            = \norm{\gamma}_{\mcL^{1, 2}} - \Tr{\gamma}
            \le \liminf_{i\to\infty} \norm{\gamma_i}_{\mcL^{1, 2}} - \Tr{\gamma} \nonumber\\
            =&\ \liminf_{i\to\infty} \prth{\norm{\gamma_i}_{\mcL^{1, 2}} - \Tr{\gamma_i}}
            = \liminf_{i\to\infty}\Tr{\gamma_i \mcM_2}. \label{eq:M2 part}
    \end{align}
    We conclude by combining \eqref{eq:kim part}, \eqref{eq:beta mass lim} with $\beta = 1$, and \eqref{eq:M2 part} in \eqref{eq:Hz/z}, using that $U \ge 0$.
\end{proof}

\subsection{Lower energy bound}\label{ch_proof_of_GS_Convergence}

We are now ready to turn to the proof of the main result.

\begin{proof}[Proof of Theorem~\ref{th:energy cvg}]

    With Proposition \ref{prop:UB} in mind we only need to gather the elements for the proof of the lower bound.

    Let $(\psi_{1, z})_{z\in \mbN}$ be a minimizing sequence of 
    \begin{align*}
        \lim_{z\to\infty} \inf_{\substack{\psi_{1, z} \in \mcF_{1, z} \\ \norm{\psi_{1, z}}_{\mcF_{1, z}}=1}} \frac{\bk{\psi_{1, z}}{H_{1, z} \psi_{1, z}}}{2z}.
    \end{align*}
    
    Taking the vacuum as a trail state, one sees that the energy is negative. Denote $\gamma_{1, z} \coloneq \ket{\psi_{1, z}} \bra{\psi_{1, z}}$, then Proposition \ref{prop:M2 esti}, with $J \ge 0$, implies that 
    \begin{align}
        \Tr{\gamma_{1, z} \mcM_2} \le 2\prth{\frac{2\mu}{U}+1}^2. \label{eq:M2 final bound}
    \end{align}
    Recalling Definition \ref{def:moments}, we find
    \begin{align*}
        \forall k\in \intint{0:z},\ \Tr{\gamma_{1, z} \mcM_2} 
            = \Tr{\gamma_{1, z}^{(1, 1)} \prth{\mcN_0^2 + \mcN_1^2}}
            = \Tr{\gamma_{1, z}^{(1, 1)} \mcM_2} = \Tr{\gamma_{1, z}^{(1, k)} \mcM_2},
    \end{align*}
    so $\gamma_{1, z}^{(1, k)}$ is uniformly bounded in $z$ inside $\mcL^{1, 2}\prth{\ell^2(\mbC) \otimes \ell^2(\mbC)^{\otimes_+ k}}$, as defined in Remark \ref{rk:N sobolev}. After a diagonal extraction, we get
    \begin{align*}
        \forall k\in \mbN,\ \gamma_{1, z}^{(1, k)} \wslim\displaylimits_{z\to\infty} \gamma^{(1, k)} \in \mcL_+^{1, 2}\prth{\ell^2(\mbC) \otimes \ell^2(\mbC)^{\otimes_+ k}}.
    \end{align*}
    The limit stays consistent in $k$. With Proposition \ref{prop:w*lsc Hz}, the mass is preserved in the limit so $\prth{\gamma^{(1, k)}}_{k\in\mbN}$ is a $(1, \infty)$-bosonic state (see Definition \ref{def:inf bosons states}). By setting $P_m \coloneq \mathds{1}_{\mcN \le m}$ for $m\in \mbN$, we verify the last assumption of Theorem~\ref{th:QDF}. We find
    \begin{align*}
        0
            \le 1 - \Tr{\gamma^{(0, 1)} P_m} 
            = \Tr{\gamma^{(0, 1)} \mathds{1}_{\mcN > m}} 
            \le \frac{\Tr{\gamma^{(0, 1)} \mcN^2}}{(m+1)^2}
            \le \frac{\Tr{\gamma_{1, z} \mcM_2}}{(m+1)^2}  \cvgto_{m\to\infty} 0
    \end{align*}
    due to \eqref{eq:M2 final bound}, which grants us the existence of
    \begin{itemize}
        \item a probability measure $\mbP \in \mcM\prth{S_{\ell^2(\mbC)},\mbR_+}$,
        \item a function $\zeta \in L^1\prth{S_{\ell^2(\mbC)}, \mcL^1_+(\ell^2(\mbC))}$ satisfying, $\mbP$-a.e., $\Tr{\zeta} = 1$,
    \end{itemize}
    such that
    \begin{align}
        \forall k\in \mbN^*,\ \gamma^{(1, k)} = \intr_{S_{\ell^2(\mbC)}} \zeta(u) \otimes p_u^{\otimes k}  \ d\mbP(u).
    \end{align}
    Let $k \in \mbN^*$, then as a consequence of Proposition \ref{prop:w*lsc Hz},
    \begin{align}
        \liminf_{z\to\infty} \frac{\Tr{\gamma_{1, z} H_{1, z}}}{z}
            = \liminf_{z\to\infty} \frac{\Tr{\gamma_{1, z}^{(1, k)} H_{1, k}}}{k}
            \ge \frac{\Tr{\gamma^{(1, k)} H_{1, k}}}{k}
            = \intr_{S_{\ell^2(\mbC)}} \frac{\Tr{\zeta(u) \otimes p_u^{\otimes k} \ H_{1, k}}}{k}  \ d\mbP(u). \label{eq:DF application}
    \end{align}
    With a final computation involving \eqref{eq:Hz},
    \begin{align*}\begin{split}
        \frac{\Tr{\zeta(u) \otimes p_u^{\otimes k} \ H_{1, k}}}{k} 
            =&\ - J \prth{\Tr{\zeta(u) a^*} \Tr{p_u a} + \Tr{\zeta(u) a} \Tr{p_u a^*}} \\
                &+ \prth{J-\mu - \frac{U}{2}}\prth{\Tr{\zeta(u) \mcN} + \Tr{p_u \mcN}}
                + \frac{U}{2} \prth{\Tr{\zeta(u) \mcN^2} + \Tr{p_u \mcN^2}}.
    \end{split}\end{align*}
    Let $\alpha_\zeta \coloneq \Tr{\zeta a}$, then
    \begin{align*}
        \abs{\overline{\alpha_\zeta(u)} \alpha_u  + \overline{\alpha_u}\alpha_\zeta(u)} \le \abs{\alpha_\zeta(u)}^2 + \abs{\alpha_u}^2,
    \end{align*}
    so recalling \eqref{mf_energy_func} and using $J \ge 0$ we find
    \begin{align}
        \frac{\Tr{\zeta(u) \otimes p_u^{\otimes k} \ H_{1, k}}}{k}  \ge 2 \inf_{\substack{\varphi \in \ell^2(\mbC) \\ \norm{\varphi}_{\ell^2} = 1}} \bk{\varphi}{h_\varphi \varphi}. \label{eq:win}
    \end{align}
    Using Proposition \ref{prop:TIR} and then \eqref{eq:DF application} and \eqref{eq:win}, we obtain the lower energy bound
    \begin{align*}
        \liminf_{z\to\infty} \inf_{\substack{\psi \in \mcF \\ \norm{\psi}_{\mcF}=1}} \frac{\bk{\psi}{H_{V_z} \psi}}{\abs{V_z}}
            \ge&\ \liminf_{z\to\infty} \inf_{\substack{\psi \in \mcF_{1, z} \\ \norm{\psi}_{\mcF_{1, z}}=1}} \frac{\bk{\psi}{H_{1, z} \psi}}{2z}
            =  \liminf_{z\to\infty} \frac{\Tr{\gamma_{1, z} H_{1, z}}}{2z} \\
            \ge&\ \intr_{S_{\ell^2(\mbC)}} \frac{\Tr{\zeta(u) \otimes p_u^{\otimes k} \ H_{1, k}}}{2k}  \ d\mbP(u)
            \ge  \intr_{S_{\ell^2(\mbC)}} \inf_{\substack{\varphi \in \ell^2(\mbC) \\ \norm{\varphi}_{\ell^2} = 1}} \bk{\varphi}{h_\varphi \varphi}  \ d\mbP(u) \\
            \ge&\ \inf_{\substack{\varphi \in \ell^2(\mbC) \\ \norm{\varphi}_{\ell^2} = 1}} \bk{\varphi}{h_\varphi \varphi}.
    \end{align*}

\end{proof}

\section{Polaron-type quantum de Finetti theorems}\label{sec_de_Finetti}

\subsection{Finite dimensional approximation}\label{sec_Finetti_finite}

Following the previous subsection, we develop a de Finetti theorem for a system with one distinct core particle and a symmetric shell of $N \in \mbN$ particles. We start by assuming that the shell Hilbert space has finite dimension, namely $\mbC^{m+1}$ with $m \in \mbN$. Let $\mcH_0$ be a separable Hilbert space representing the core particle. This setting is similar to the one of the polaron model where one has an impurity in a bath of indistinguishable particles. This is why we expect that our methods might be applicable to proving mean-field limits for such models as well.

We define the symmetric projection
\begin{align*}
    \Pi_N^+ \coloneq \prth{\frac{1}{N!}\sum_{\sigma\in S_N} U_\sigma} : \prth{\mbC^{m+1}}^{\otimes N} \to \prth{\mbC^{m+1}}^{\otimes_+ N},
\end{align*}
where $S_N$ is the permutation group on $N$ elements and $U_\sigma$ is the unitary defined by
\begin{align*}
    \forall \psi \in \prth{\mbC^{m+1}}^{\otimes N},\ x_{1:N} \in \prth{\mbC^{m+1}}^N,\ U_\sigma \psi (x_{1:N}) \coloneq \psi\prth{x_{\sigma(1)}, \dots, x_{\sigma(N)}}.
\end{align*}

Denote by $S^m \subseteq \mbC^{m+1}$ the complex $m$-sphere, with $h_m$ being the normalized Haar measure on $S^m$. We recall the Schur formula:
\begin{align}
    \Pi_N^+ = \prth{\matrx{N+m \\ m}}\intr_{S^m} p_u^{\otimes N} \, dh_m(u), \label{eq:Schur}
\end{align}
where $p_u$ is the orthogonal projection onto $u \in \mbC^{m+1}$. This is a consequence of Schur's lemma applied to the following irreducible representation:
\begin{align*}
    \matrx{&\mathcal{U}_{m+1}(\mbC)&\to &\textnormal{End}\prth{(\mbC^{m+1})^{\otimes_+ N}} \\ &U&\mapsto& U^{\otimes N}.}
\end{align*}
Our first result is the following.

\begin{theorem}[Polaron quantum de Finetti for finite $N$]\label{th:finite DF}
    Let $\gamma_{1, N} \in \mcL\prth{\mcH_0 \otimes\prth{\mbC^{m+1}}^{\otimes_+ N}}$ be a $(1, N)$-bosonic state (see Definition \ref{def:mult bosons}). Then
    \begin{align}
        \eta_{1, N} \coloneq \prth{\matrx{N+m \\ m}}\intr_{S^m} \prth{\mathds{1}_{\mcH_0}\otimes p_u^{\otimes N} \ \gamma_{1, N}}^{(1, 0)}  \otimes p_u^{\otimes N} dh_m(u) \in \mcL_+^1\prth{\mcH_0 \otimes\prth{\mbC^{m+1}}^{\otimes_+ N}} \label{eq:eta def}
    \end{align}
    satisfies $\Tr{\eta_{1, N}} = 1$ and
    \begin{align}
        \norm{\gamma_{1, N}^{(1, k)} - \eta_{1, N}^{(1, k)}}_{\mcL^1} \le \frac{4mk}{N+1}. \label{eq:DF trace estimate}
    \end{align}
\end{theorem}
\begin{proof}
    First, $\eta_{1, N}$ is positive as a sum of positive operators. Then, using \eqref{eq:Schur},
    \begin{align}
        \Tr{\eta_{1, N}} 
            = \prth{\matrx{N+m \\ m}}\intr_{S^m} \Tr{\mathds{1}_{\mcH_0}\otimes p_u^{\otimes N} \ \gamma_{1, N}} dh_m(u)
            = \Tr{\mathds{1}_{\mcH_0} \otimes \Pi_N^+\ \gamma_{1, N}}
            = \Tr{\gamma_{1, N}}
            = 1. \label{eq:tr eta}
    \end{align}
    With the notation $\mathds{1} \coloneq \mathds{1}_{\mbC^{m+1}}$, we start by expressing $\gamma_{1, N}^{(1, k)}$ using \eqref{eq:Schur}:
    \begin{align}
        \gamma_{1, N}^{(1, k)}
            = \prth{\mathds{1}_{\mcH_0} \otimes \mathds{1}^{\otimes k} \otimes \Pi_{N-k}^+\ \gamma_{1, N}}^{(1, k)} 
            = \prth{\matrx{N-k+m \\ m}}\intr_{S^m} \underbrace{\prth{\mathds{1}_{\mcH_0} \otimes \mathds{1}^{\otimes k} \otimes p_u^{\otimes (N-k)} \gamma_{1, N}}^{(1, k)}}_{\eqcolon \gamma_k(u)\in\mcL_+^1\prth{\mcH_0\otimes \prth{\mbC^{m+1}}^{\otimes_+ k}}}dh_m(u). \label{eq:gammaNk}
    \end{align}
    We claim that
    \begin{align}
        \mathds{1}_{\mcH_0}\otimes p_u^{\otimes k}\ \gamma_k(u)\ \mathds{1}_{\mcH_0}\otimes p_u^{\otimes k} 
            = \gamma_0(u) \otimes p_u^{\otimes k}. \label{eq:core idea}
    \end{align}
    Indeed, if $A \in \mathcal{K}(\mcH_0)$ and $B\in \mathcal{K}\prth{\prth{\mbC^{m+1}}^{\otimes k}}$,
    \begin{align*}
        \Tr{A \otimes B \ \mathds{1}_{\mcH_0}\otimes p_u^{\otimes k}\ \gamma_k(u)\ \mathds{1}_{\mcH_0}\otimes p_u^{\otimes k}}
            =&\ \Tr{A \otimes p_u^{\otimes k} B p_u^{\otimes k} \ \gamma_k(u)} \\
            =&\ \Tr{A \otimes p_u^{\otimes k} \ \gamma_k(u)} \Tr{p_u^{\otimes k} B} \\
            =&\ \Tr{A \otimes p_u^{\otimes N} \ \gamma_{1, N}} \Tr{p_u^{\otimes k} B} 
            = \Tr{A \gamma_0(u)} \Tr{p_u^{\otimes k} B} \\
            =&\ \Tr{A\otimes B \ \gamma_0(u) \otimes p_u^{\otimes k}}.
    \end{align*}
    Inserting \eqref{eq:core idea} into \eqref{eq:eta def},
    \begin{align}
        \eta_{1, N}^{(1, k)} 
            =&\ \prth{\matrx{N+m \\ m}} \intr_{S^m} \gamma_0(u)\otimes p_u^{\otimes k} dh_m(u)
            = \prth{\matrx{N+m \\ m}} \intr_{S^m} \mathds{1}_{\mcH_0}\otimes p_u^{\otimes k}\ \gamma_k(u)\ \mathds{1}_{\mcH_0}\otimes p_u^{\otimes k} dh_m(u) \nonumber\\
            =&\ \prth{\matrx{N-k+m \\ m}} \intr_{S^m} \mathds{1}_{\mcH_0}\otimes p_u^{\otimes k}\ \gamma_k(u)\ \mathds{1}_{\mcH_0}\otimes p_u^{\otimes k} dh_m(u) + \prth{1 - \frac{\prth{\matrx{N-k+m\\m}}}{\prth{\matrx{N+m\\m}}}} \eta_{1, N}^{(1, k)}. \label{eq:gammaNktilde}
    \end{align}
    Moreover, using \eqref{eq:Schur},
    \begin{align}
        \prth{\matrx{N+m\\m}}\intr_{S_m} \mathds{1}_{\mcH_0} \otimes p_u^{\otimes k} \ \gamma_k(u) dh_m(u)
            = \prth{\matrx{N+m\\m}} \intr_{S_m} \prth{\mathds{1}_{\mcH_0}\otimes p_u^{\otimes N} \ \gamma_{1, N}}^{(1, k)}dh_m(u)
            = \gamma_{1, N}^{(1, k)}. \label{eq:gammaNkk}
    \end{align}
    As $\gamma_k(u) \ge 0$, 
    \begin{align*}
        &\norm{\ \intr_{S^m} \mathds{1}_{\mcH_0}\otimes (1-p_u^{\otimes k})\ \gamma_k(u) \ \mathds{1}_{\mcH_0}\otimes (1-p_u^{\otimes k}) dh_m(u)}_{\mcL^1} \\
            =&\ \Tr{\ \intr_{S^m} \mathds{1}_{\mcH_0}\otimes (1-p_u^{\otimes k})\ \gamma_k(u) \ \mathds{1}_{\mcH_0}\otimes (1-p_u^{\otimes k}) dh_m(u)}\\
            =&\ \Tr{\ \intr_{S^m} \mathds{1}_{\mcH_0}\otimes (1-p_u^{\otimes k})\ \gamma_k(u) dh_m(u)} 
            \le \norm{\ \intr_{S^m} \mathds{1}_{\mcH_0}\otimes (1-p_u^{\otimes k})\ \gamma_k(u) dh_m(u)}_{\mcL^1}.
    \end{align*}
    Hence, combining
    \begin{align*}
        &\gamma_k(u) - \mathds{1}_{\mcH_0}\otimes p_u^{\otimes k}\ \gamma_k(u)\ \mathds{1}_{\mcH_0}\otimes p_u^{\otimes k} \\
            =&\ \mathds{1}_{\mcH_0}\otimes (1-p_u^{\otimes k})\ \gamma_k(u) + \gamma_k(u) \ \mathds{1}_{\mcH_0}\otimes (1-p_u^{\otimes k})  - \mathds{1}_{\mcH_0}\otimes (1-p_u^{\otimes k})\ \gamma_k(u) \ \mathds{1}_{\mcH_0}\otimes (1-p_u^{\otimes k}),
    \end{align*}
    with \eqref{eq:gammaNk} and \eqref{eq:gammaNkk}, we obtain
    \begin{align}
        &\norm{\intr_{S_m} \prth{\gamma_k(u) - \mathds{1}_{\mcH_0}\otimes p_u^{\otimes k}\ \gamma_k(u)\ \mathds{1}_{\mcH_0}\otimes p_u^{\otimes k} }dh_m(u)}_{\mcL^1} 
            \le 3 \norm{\intr_{S^m} \mathds{1}_{\mcH_0}\otimes (1-p_u^{\otimes k})\ \gamma_k(u) dh_m(u)}_{\mcL^1} \nonumber\\
            =&\ 3\norm{\prth{\prth{\matrx{N-k+m\\m}}^{-1} - \prth{\matrx{N+m\\m}}^{-1}}\gamma_{1, N}^{(1, k)}}_{\mcL^1}
            = 3\prth{\prth{\matrx{N-k+m\\m}}^{-1} - \prth{\matrx{N+m\\m}}^{-1}}. \label{eq:square diff}
    \end{align}
    With \eqref{eq:gammaNk} and \eqref{eq:gammaNktilde}, 
    \begin{align*}
        \gamma_{1, N}^{(1, k)} - \eta_{1, N}^{(1, k)} 
            =&\ \prth{\matrx{N-k+m \\ m}}\intr_{S_m} \prth{\gamma_k(u) - \mathds{1}_{\mcH_0}\otimes p_u^{\otimes k}\ \gamma_k(u)\ \mathds{1}_{\mcH_0}\otimes p_u^{\otimes k} }dh_m(u) \\
                &+ \prth{1 - \frac{\prth{\matrx{N-k+m\\m}}}{\prth{\matrx{N+m\\m}}}} \eta_{1, N}^{(1, k)}.
    \end{align*}
    Then, inserting \eqref{eq:square diff} and \eqref{eq:tr eta} along with
    \begin{align*}
         1&\ge \frac{\prth{\matrx{N-k+m\\m}}}{\prth{\matrx{N+m\\m}}}
            =
         \frac{(N-k+m)! N!}{(N-k)!(N+m)!}
            = \prod_{i=1}^m \frac{N-k+i}{N+i}
            \ge \prth{\frac{N-k+1}{N+1}}^{m} 
            = \prth{1-\frac{k}{N+1}}^{m} \\ 
            &\ge 1 -  \frac{mk}{N+1},
    \end{align*}
    we get
    \begin{align}
        \norm{\gamma_{1, N}^{(1, k)} - \eta_{1, N}^{(1, k)}}_{\mcL^1}
            \le 4 \prth{1 - \frac{\prth{\matrx{N-k+m\\m}}}{\prth{\matrx{N+m\\m}}}}
            \le 4\frac{mk}{N+1}.  \label{eq:L1 esti}
    \end{align}
\end{proof}
\begin{remark}
    Some comments:
    \begin{itemize}
        \item Our intuition is that \eqref{eq:core idea} is the main novelty compared to the usual de Finetti argument. \emph{A priori}, we only expect factorization to occur between the many symmetric variables, and not between the impurity and the rest. However, since we are able to project the symmetric part onto a rank-one projection $p_u^{\otimes k}$, the resulting $(1,k)$-variable operators exhibit a full factorization.
        \item De Finetti theorems have already been used for multi-species systems \cite{MNO}, with a large number of symmetric particles in each species. Our De Finetti theorem can also be generalized to this setting: one impurity and multiple species with a large number of particles in each species.
    \end{itemize}
\end{remark}

We now consider the limit $N \to \infty$ while keeping the dimension $m+1$ of the Hilbert space of the symmetric particles fixed. Note the following.

\begin{definition}[$S^1$-invariance] \label{def:S1 inv}
    Let $\zeta$ be a trace-class valued functions and $\mbP$ a probability measure on a Hilbert space. $(\zeta, \mbP)$ is said to be $S^1$ invariant if 
    \begin{align*}
        \forall \theta \in \mbR,\ e^{i\theta}_* \mbP = \mbP \text{ and } \mbP\text{-a.e.}, \zeta = \zeta \circ e^{i\theta},
    \end{align*}
    where the phase $e^{i\theta}$ is set to act multiplicatively on every complex coordinates of the Hilbert space.
\end{definition}

We denote by $\mcM$ the set of Radon measures and $C^0$ spaces of continuous functions.

\begin{theorem}[Polaron quantum de Finetti for $N=\infty$] \label{th:QDFm}
        Let $\gamma_m \in \mcL\prth{\mcH_0 \otimes  \mcF_+\prth{\mbC^{m+1}}}$ be a $(1,\infty)$-bosonic state. Then there exist a unique $S^1$-invariant
        \begin{itemize}
            \item probability measure $\mbP_m \in \mcM\prth{S^m,\mbR_+}$,
            \item function $\zeta_m \in L^1\prth{S^m, \mcL^1_+(\mcH_0)}$ satisfying, $\mbP_m$-a.e., $\Tr{\zeta_m} = 1$,
        \end{itemize}
        such that
        \begin{align}
            \forall k\in \mbN,\ \gamma_m^{(1, k)} = \intr_{S^m} \zeta_m(u)\otimes p_u^{\otimes k}  \ d\mbP_m(u). \label{eq:limit DF}
        \end{align}
        Moreover, up to a subsequence, we have the following approximation:
        \begin{align}
            \prth{\matrx{N+m \\ m}}\prth{\mathds{1}_{\mcH_0}\otimes p_\bullet^{\otimes N} \ \gamma_m^{(1, N)}}^{(1, 0)} h_m \wslim\displaylimits_{N\to\infty} \zeta_m \mbP_m  \text{ as trace-class valued Radon measures.}\label{eq:zeta P cvg}
        \end{align}
\end{theorem}

\begin{proof}
    Let $N\in\mbN$, we define
    \begin{align}
        \zeta_{N, m}: \matrx{S^m&\to&\mcL_+^1\prth{\mcH_0} \\ u&\mapsto & \prth{\matrx{N+m \\ m}}\prth{\mathds{1}_{\mcH_0}\otimes p_u^{\otimes N} \ \gamma_m^{(1, N)}}^{(1, 0)}} \label{eq:zNM}
    \end{align}
    and observe that, using \eqref{eq:Schur},
    \begin{align*}
        \norm{\zeta_{N, m}}_{L^1} 
            =&\ \intr_{S^m} \norm{\zeta_{N, m}(u)}_{\mcL^1}dh_m(u)
            = \prth{\matrx{N+m \\ m}} \intr_{S^m} \Tr{\mathds{1}_{\mcH_0}\otimes p_u^{\otimes N} \ \gamma_m^{(1, N)}} dh_m(u) \\
            =&\ \Tr{\mathds{1}_{\mcH_0}\otimes\Pi_N^+ \ \gamma_m^{(1, N)}}
            = \Tr{\gamma_m^{(1, N)}}
            = 1.
    \end{align*}
    Thus
    \begin{align*}
         \prth{\zeta_{N, m} h_m}_N \subseteq  \mcM\prth{S^m, \mcL^1(\mcH_0)} = C^0\prth{S^m, \mcK(\mcH_0)}^*
    \end{align*}
    is a bounded sequence of trace-class-valued Radon measures. By the Banach--Alaoglu theorem, up to a subsequence,
    \begin{align}
        \zeta_{N, m} h_m \wslim\displaylimits_{N\to\infty} Z_m \in \mcM\prth{S^m, \mcL_+^1(\mcH_0)}. \label{eq:zeta lim}
    \end{align}
    Moreover,
    \begin{align}
        \mbP_m \coloneq \norm{Z_m}_{\mcL^1} = \Tr{Z_m} \in \mcM\prth{S^m,\mbR_+}  \label{eq:def Pm}
    \end{align}
    has bounded total variation since
    \begin{align*}
        \Tr{Z_m(S^m)} \le \liminf_{N\to\infty} \norm{\zeta_{N, m}}_{L^1} = 1.
    \end{align*}
    As $Z_m \ll \mbP_m$, by the Radon--Nikodym theorem,
    \begin{align*}
        \exists \zeta_m \in L^1\prth{S^m, \mcL^1_+(\mcH_0)} \text{ such that } Z_m = \zeta_m \mbP_m,
    \end{align*}
    and \eqref{eq:zeta lim} implies \eqref{eq:zeta P cvg}. Taking the trace, it follows from \eqref{eq:def Pm} that
    \begin{align*}
        \mbP_m = \Tr{\zeta_m} \mbP_m,
    \end{align*}
    meaning that $\mbP_m$-a.e., $\Tr{\zeta_m} = 1$. Then \eqref{eq:zeta lim} has the following meaning:
    \begin{align}
        \forall \varphi \in C^0\prth{S^m, \mcK(\mcH_0)},\ \intr_{S^m} \Tr{\zeta_{N, m}(u)\varphi(u)} dh_m(u)
            \cvgto\displaylimits_{N\to\infty} \intr_{S^m}\Tr{\zeta_m(u)\varphi(u)}d\mbP_m(u). \label{eq:wcvg fn}
    \end{align}
    With the same notation as in Theorem~\ref{th:finite DF} applied to $\gamma_m^{(1, N)}$, we define
    \begin{align*}
        \eta_{1, N} \coloneq\intr_{S^m} \zeta_{N, m}(u) \otimes p_u^{\otimes N} dh_m(u).
    \end{align*}
    Let $k\in\mbN$ and $K \in \mcK\prth{\mcH_0 \otimes\prth{\mbC^{m+1}}^{\otimes_+ k}}$. Considering the test function
    \begin{align*}
        u \mapsto \prth{\mathds{1}_{\mcH_0} \otimes p_u^{\otimes k}\ K}^{(1, 0)} \in C^0\prth{S^m, \mcK(\mcH_0)},
    \end{align*}
    in \eqref{eq:wcvg fn}, we obtain
    \begin{align*}
        \Tr{\eta_{1, N}^{(1, k)} K} 
            =&\ \intr_{S^m} \Tr{\zeta_{N, m}(u) \otimes p_u^{\otimes k} \ K}dh_m(u)
            = \intr_{S^m} \Tr{\zeta_{N, m}(u) \prth{\mathds{1}_{\mcH_0} \otimes p_u^{\otimes k}\ K}^{(1, 0)} }dh_m(u) \\
            \cvgto\displaylimits_{N\to\infty}&\ \intr_{S^m} \Tr{\zeta_m(u)\prth{\mathds{1}_{\mcH_0} \otimes p_u^{\otimes k}\ K}^{(1, 0)}} d\mbP_m(u)
            = \intr_{S^m} \Tr{\zeta_m(u) \otimes p_u^{\otimes k}\ K} d\mbP_m(u) \\
            =&\ \Tr{\prth{\ \intr_{S^m} \zeta_m(u) \otimes p_u^{\otimes k}  \ d\mbP_m(u)} K}.
    \end{align*}
    Hence,
    \begin{align*}
        \eta_{1, N}^{(1, k)} \wslim\displaylimits_{N\to\infty} \intr_{S^m} \zeta_m(u) \otimes p_u^{\otimes k}  \ d\mbP_m(u).
    \end{align*}
    Using \eqref{eq:DF trace estimate}, and observing that
    \begin{align*}
        \norm{\eta_{1, N}^{(1, k)} - \gamma_m^{(1, k)}}_{\mcL^1}
            =\norm{\eta_{1, N}^{(1, k)} - \prth{\gamma_m^{(1, N)}}^{(1, k)}}_{\mcL^1}
            \cvgto_{N\to\infty} 0,
    \end{align*}
    we conclude that
    \begin{align*}
        \gamma_m^{(1, k)} 
            = \intr_{S^m} \zeta_m(u) \otimes p_u^{\otimes k}  \ d\mbP_m(u).
    \end{align*}
    Finally, taking the trace,
    \begin{align*}
        1 = \Tr{\gamma_m^{(1, k)}} 
            = \intr_{S^m} \Tr{\zeta_m(u)} \Tr{p_u}^k \, d\mbP_m(u)
            = \intr_{S^m} d\mbP_m(u)
            = \mbP_m(S^m),
    \end{align*}
    so $\mbP_m$ is a probability measure.

\textbf{$S^1$-Invariance.}
    From \eqref{eq:zeta P cvg} we infer that $\zeta_m \mbP_m$ is $S^1$-invariant as a trace-class-valued measure, being the limit of such measures. Taking the trace, this implies that $(\zeta_m, \mbP_m)$ is $S^1$-invariant in the sense of Definition \ref{def:S1 inv}.
    
    Tracing out the first variable in \eqref{eq:limit DF}, we reduce the situation to the known case where $\mathds{P}_m$ is the unique $S^1$-invariant probability measure on $S^m$ satisfying $\eqref{eq:limit DF}$. Then, if there exists another candidate $\widetilde{\zeta}_m \in L^1\prth{S^m, \mcL^1_+(\mcH_0)}$, we find that 
    \begin{align*}
        \forall k \in \mbN,\ \intr_{S^m} \prth{\zeta_m(u) - \widetilde{\zeta}_m(u)} \otimes p_u^{\otimes k} \ d\mbP_m(u) = 0.
    \end{align*}
    Using the standard arguments for the uniqueness of $\mbP_m$, i.e., testing the above against $\mathds{1}_{\mcH_0} \otimes \bigotimes_{i=1}^k A_i$ for $A_{1:k} \in \mcL\prth{\mbC^{m+1}}$ self-adjoint and $k\in \mbN$, we find, by density (Stone--Weierstrass theorem) of the algebra generated by functions of the form
    \begin{align*}
        u \mapsto \prod_{i=1}^k \bk{u}{A_i u} \in C^0\prth{\quotient{S^m}{S^1}, \mbR},
    \end{align*}
    that $\forall f \in C^0\prth{\quotient{S^m}{S^1}, \mbR}$,
    \begin{align*}
        \intr_{S^m} \prth{\zeta_m(u) - \widetilde{\zeta}_m(u)} f(u) \ d\mbP_m(u) = 0.
    \end{align*}
    This implies that
    \begin{align*}
        \prth{\zeta_m - \widetilde{\zeta}_m} \mbP_m = 0
    \end{align*}
    as a trace-class-valued measure, meaning that $\mbP_m$-a.e., $\zeta_m = \widetilde{\zeta}_m$. 
\end{proof}

\subsection{Fock space localization}\label{sec_Fock_space_loc}

In order to deal with the finite-dimensional approximation in \eqref{th:QDFm}, and replace $\mbC^{m+1}$ by a separable Hilbert space $\mcH_s$, we use the Fock space localization method, see \cite{LEWIN2014570}.

\begin{proposition}[Fock space localization]\label{prop:loc}
    Let $\gamma_{1, N} \in \mcL\prth{\mcH_0 \otimes\prth{\mcH_s}^{\otimes_+ N}}$ be a $(1, N)$-bosonic state, $P$ be an orthonormal projection on $\mcH_s$, and $Q \coloneq \mathds{1}_{\mcH_s} - P$. The trace-class valued measure
    \begin{align*}
        M_{N, k} \coloneq \sum_{n=k}^N \prth{\matrx{N \\ n}}\prth{\mathds{1}_{\mcH_0}\otimes P^{\otimes n} \otimes Q^{\otimes(N-n)} \ \gamma_{1, N} \ \mathds{1}_{\mcH_0}\otimes P^{\otimes n} \otimes Q^{\otimes(N-n)}}^{(1, k)}\delta_{\frac{n}{N}}
    \end{align*}
    satisfies
    \begin{align}
        \forall k \in \mbN,\ 
        \norm{\mathds{1}_{\mcH_0} \otimes P^{\otimes k} \ \gamma_{1, N}^{(1, k)} \ \mathds{1}_{\mcH_0} \otimes P^{\otimes k} - \intr_0^1 \lambda^k \, dM_{N, k}(\lambda)}_{\mcL^1}
            \le \frac{k(k-1)}{N}. \label{eq:loc limit}
    \end{align}
\end{proposition}
\begin{proof}
    Using the symmetry of $\gamma_{1, N}$ with respect to its last $N$ variables, we may apply the commutative binomial formula and observe that, for $k \in \intint{0, N}$,
    \begin{align*}
        \gamma_{1, N}^{(1, k)} 
            =&\ \prth{\mathds{1}_{\mcH_0}\otimes \mathds{1}_{\mcH_s}^{\otimes k}\otimes (P+Q)^{\otimes(N-k)} \ \gamma_{1, N}}^{(1, k)} \\
            =&\ \sum_{n=0}^{N-k} \prth{\matrx{N-k \\ n}} \prth{\mathds{1}_{\mcH_0}\otimes \mathds{1}_{\mcH_s}^{\otimes k}\otimes P^{\otimes n} \otimes Q^{\otimes(N-k-n)} \ \gamma_{1, N}}^{(1, k)} \\ 
            =&\ \sum_{n=k}^N \prth{\matrx{N-k \\ n-k}} \prth{\mathds{1}_{\mcH_0}\otimes \mathds{1}_{\mcH_s}^{\otimes k}\otimes P^{\otimes(n-k)} \otimes Q^{\otimes(N-n)} \ \gamma_{1, N}}^{(1, k)}.
    \end{align*}
    Therefore,
    \begin{align}
        &\mathds{1}_{\mcH_0} \otimes P^{\otimes k} \ \gamma_{1, N}^{(1, k)} \ \mathds{1}_{\mcH_0} \otimes P^{\otimes k}\nonumber \\
            =&\ \sum_{n=k}^N \prth{\matrx{N-k \\ n-k}} \prth{\mathds{1}_{\mcH_0}\otimes P^{\otimes n} \otimes Q^{\otimes(N-n)} \ \gamma_{1, N} \ \mathds{1}_{\mcH_0}\otimes P^{\otimes n} \otimes Q^{\otimes(N-n)}}^{(1, k)}. \label{eq:loc N}
    \end{align}
    Combining \eqref{eq:loc N} with the definition of $M_{N,k}$ yields
    \begin{align}
        &\norm{\mathds{1}_{\mcH_0} \otimes P^{\otimes k} \ \gamma_{1, N}^{(1, k)} \ \mathds{1}_{\mcH_0} \otimes P^{\otimes k} - \intr_0^1 \lambda^k \, dM_{N, k}(\lambda)}_{\mcL^1} \nonumber \\
            \le&\ \sum_{n=k}^N \abs{\prth{\matrx{N-k \\ n-k}}- \prth{\matrx{N \\ n}}\prth{\frac{n}{N}}^k}\Tr{\mathds{1}_{\mcH_0}\otimes P^{\otimes n} \otimes Q^{\otimes(N-n)} \ \gamma_{1, N}}. \label{eq:loc estimate}
    \end{align}
    Assume that $k \ge 1$ and let $n \in \intint{k, N}$. We estimate
    \begin{align*}
        \prth{\matrx{N \\ n}}^{-1}\prth{\matrx{N-k \\ n-k}}
            =&\ \frac{n! (N-k)!}{N!(n-k)!}
            = \prod_{j=0}^{k-1} \frac{n-j}{N-j}
            \ge \prth{\frac{n-(k-1)}{N}}^k \\
            =&\ \prth{\frac{n}{N}}^k\prth{1-\frac{k-1}{n}}^k
            \ge \prth{\frac{n}{N}}^k \prth{1 - \frac{k(k-1)}{n}}.
    \end{align*}
    Hence,
    \begin{align*}
        0 
            \le \prth{\frac{n}{N}}^k - \prth{\matrx{N \\ n}}^{-1}\prth{\matrx{N-k \\ n-k}}
            \le \prth{\frac{n}{N}}^k \frac{k(k-1)}{n}
            \le \frac{k(k-1)}{N}.
    \end{align*}
    If $k=0$, the right-hand side in \eqref{eq:loc estimate} vanishes. Otherwise, it is bounded by
    \begin{align*}
        &\norm{\mathds{1}_{\mcH_0} \otimes P^{\otimes k} \ \gamma_{1, N}^{(1, k)} \ \mathds{1}_{\mcH_0} \otimes P^{\otimes k} - \intr_0^1 \lambda^k \, dM_{N, k}(\lambda)}_{\mcL^1} \\
            \le&\ \frac{k(k-1)}{N}\sum_{n=k}^N \prth{\matrx{N \\ n}}\Tr{\mathds{1}_{\mcH_0}\otimes P^{\otimes n} \otimes Q^{\otimes(N-n)} \ \gamma_{1, N}}
            \le \frac{k(k-1)}{N}.
    \end{align*}
\end{proof}

We are now ready to prove Theorem~\ref{th:QDF}. Let us denote by $S_{\mcH_s}$ the sphere in $\mcH_s$. The core idea is that $P_m$ acts on the symmetric variables, thereby reducing the statement to the finite-dimensional case.

\begin{proof}[Proof of Theorem~\ref{th:QDF}]
    Let $N, m\in\mbN$, $k\in \intint{0,N}$, and with the usual notation,
    \begin{align*}
        Q_m \coloneq \mathds{1}_{\mcH_s} - P_m.
    \end{align*}
    In view of applying Proposition \ref{prop:loc} to $\gamma^{(1, N)}$, we define
    \begin{align}
        M_{N, k, m} \coloneq \sum_{n=k}^N \underbrace{\prth{\matrx{N \\ n}}\prth{\mathds{1}_{\mcH_0}\otimes P_m^{\otimes n} \otimes Q_m^{\otimes(N-n)} \ \gamma^{(1, N)} \ \mathds{1}_{\mcH_0}\otimes P_m^{\otimes n} \otimes Q_m^{\otimes(N-n)}}^{(1, k)}}_{\coloneq \eta_{N, k} \in \mcL^1_+\prth{\mcH_0 \otimes (P_m \mcH_s)^{\otimes_+ k}}}\delta_{\frac{n}{N}}. \label{eq:Mnkm}
    \end{align}

\subProof{Taking $N\to\infty$}

    $\prth{M_{N, k, m}}_{N\in\mbN}$ is bounded in $\mcM\prth{\sbra{0, 1}, \mcL^1_+\prth{\mcH_0 \otimes (P_m \mcH_s)^{\otimes_+ k}}}$:
    \begin{align}
        \Tr{M_{N, k, m}(\sbra{0, 1})} 
            =& \sum_{n=k}^N \Tr{\eta_{N, k}}
            = \sum_{n=k}^N \prth{\matrx{N \\ n}}\Tr{\mathds{1}_{\mcH_0}\otimes P^{\otimes n} \otimes Q^{\otimes(N-n)} \ \gamma^{(1, N)}} \nonumber\\
            \le&\ \sum_{n=0}^N \prth{\matrx{N \\ n}}\Tr{\mathds{1}_{\mcH_0}\otimes P^{\otimes n} \otimes Q^{\otimes(N-n)} \ \gamma^{(1, N)}}
            = \Tr{\gamma^{(1, N)}}
            = 1. \label{eq:Tr sum eta}
    \end{align}
    Hence, after extraction,
    \begin{align}
        M_{N, k, m} \wslim\displaylimits_{N\to\infty} M_{k, m} \in \mcM\prth{\sbra{0, 1}, \mcL^1_+\prth{\mcH_0 \otimes (P_m \mcH_s)^{\otimes_+ k}}}. \label{eq:Mnkm w* lim}
    \end{align}
    Let $K \in \mcK\prth{\mcH_0 \otimes (P_m \mcH_s)^{\otimes_+ k}}$. Then $\lambda \mapsto \lambda^k K \in C^0\prth{\sbra{0, 1}, \mcK\prth{\mcH_0 \otimes (P_m \mcH_s)^{\otimes_+ k}}}$, and thus
    \begin{align*}
        \Tr{K \intr_0^1 \lambda^k dM_{N, k, m}(\lambda)}
        \cvgto_{N\to\infty} \Tr{K\intr_0^1 \lambda^k dM_{k, m}(\lambda)}.
    \end{align*}
    It follows that
    \begin{align*}
        \intr_0^1\lambda^k dM_{N, k, m}(\lambda) \wslim\displaylimits_{N\to\infty}\intr_0^1 \lambda^k dM_{k, m}(\lambda) \in \mcL^1\prth{\mcH_0 \otimes (P_m \mcH_s)^{\otimes_+ k}}.
    \end{align*}
    Together with \eqref{eq:loc limit}, we obtain
    \begin{align}
        \mathds{1}_{\mcH_0} \otimes P_m^{\otimes k} \ \gamma^{(1, k)} \ \mathds{1}_{\mcH_0} \otimes P_m^{\otimes k} = \intr_0^1 \lambda^k dM_{k, m}(\lambda), \label{eq:localization}
    \end{align}
    since $\prth{\gamma^{(1, N)}}^{(1, k)} = \gamma^{(1, k)}$ is now independent of $N$.

\subProof{Consistency of $M_{k, m}$}

    We can canonically extend Definitions \ref{def:mult bosons} and \ref{def:inf bosons states} to trace-class valued measures. Since $P_m \mcH_s$ is finite-dimensional, we can trace out one coordinate in \eqref{eq:Mnkm w* lim} to obtain
    \begin{align}
        M_{N, k+1, m}^{(1, k)} \wslim\displaylimits_{N\to\infty} M_{k+1, m}^{(1, k)}. \label{eq:k+1 cvg}
    \end{align}
    Following from \eqref{eq:Mnkm},
    \begin{align}
        M_{N, k+1, m}^{(1, k)} = M_{N, k, m} -  \eta_{N, k} \delta_{\frac{k}{N}}. \label{eq:approx cstc}
    \end{align}
    The estimate \eqref{eq:Tr sum eta} implies that $\norm{\eta_{N, k}}_{\mcL^1} = \Tr{\eta_{N, k}} \le 1$, so we can extract a limit
    \begin{align*}
        \eta_{N, k} \wslim\displaylimits_{N\to\infty} \eta_k \in \mcL^1_+\prth{\mcH_0 \otimes (P_m \mcH_s)^{\otimes_+ k}}.
    \end{align*}
    Let $\varphi \in C^0\prth{\sbra{0, 1}, \mcK\prth{\mcH_0 \otimes (P_m \mcH_s)^{\otimes_+ k}}}$. Then
    \begin{align*}
        \intr_0^1 \Tr{\eta_{N, k} \varphi(\lambda)}d\delta_{\frac{k}{N}}
            =&\ \Tr{\eta_{N, k} \varphi\prth{\frac{k}{N}}}
            = \Tr{\eta_{N, k} \varphi(0)} + \Tr{\eta_{N, k} \prth{\varphi\prth{\frac{k}{N}} - \varphi(0)}} \\
            \cvgto_{N\to\infty}&\ \Tr{\eta_k \varphi(0)},
    \end{align*}
    since $\varphi(0)$ is compact and $\norm{\varphi\prth{\frac{k}{N}} - \varphi(0)}_{\mcL^\infty} \cvgto\limits_{N\to\infty} 0$ by continuity of $\varphi$. We deduce that
    \begin{align*}
        \eta_{N, k} \delta_{\frac{k}{N}} \wslim\displaylimits_{N\to\infty} \eta_k \delta_0 \in \mcM\prth{\sbra{0, 1}, \mcL^1_+\prth{\mcH_0 \otimes (P_m \mcH_s)^{\otimes_+ k}}}.
    \end{align*}
    Together with \eqref{eq:Mnkm w* lim} and \eqref{eq:k+1 cvg}, taking $N \to\infty$ in \eqref{eq:approx cstc} yields
    \begin{align}
        M_{k+1, m}^{(1, k)} = M_{k, m} -  \eta_{k} \delta_0. \label{eq:almost cst}
    \end{align}
    Keeping in mind that $M_{k, m}$ will be integrated against $\lambda^k$ (see \eqref{eq:localization}), we do not care about the mass at $0$. Hence, let
    \begin{align*}
        \forall B \in \mcB(\sbra{0, 1}),\ \widetilde{M}_{k, m}(B) = M_{k, m}\prth{B\backslash \sett{0}} + M_{0, m}\prth{B \cap \sett{0}}\otimes \prth{\frac{\mathds{1}_{\mbC^{m+1}}}{m+1}}^{\otimes k},
    \end{align*}
    so that $\widetilde{M}_{0, m} = M_{0, m}$ and, with \eqref{eq:almost cst} and then \eqref{eq:Tr sum eta},
    \begin{align}
        \Tr{\widetilde{M}_{k, m}[0, 1]} 
            =&\ \Tr{M_{k, m}\prth{(0, 1]}} + \Tr{M_{0, m}\prth{\sett{0}}}
            = \Tr{M_{0, m}\prth{(0, 1]}} + \Tr{M_{0, m}\prth{\sett{0}}} \nonumber \\
            =&\ \Tr{M_{0, m}\prth{[0, 1]}}
            \le \liminf_{N\to\infty} \Tr{M_{N, k, m}(\sbra{0, 1})}
            \le 1, \label{eq:Mkm mass}
    \end{align}
    and $(\widetilde{M}_{k, m})_{k\in \mbN}$ is consistent in $k$:
    \begin{align}
         \widetilde{M}_{k+1, m}^{(1, k)} = \widetilde{M}_{k, m}. \label{eq:cstc tilde}
    \end{align}

\subProof{Applying Theorem~\ref{th:QDFm} to $M_{k, m}$}

    By the Radon--Nikodym theorem, there exists $\gamma_{k, m} \in L^1\prth{[0, 1], \mcL^1_+\prth{\mcH_0 \otimes (P_m \mcH_s)^{\otimes_+ k}}}$ such that
    \begin{align}
        \widetilde{M}_{k, m} 
            = \gamma_{k, m} \Tr{\widetilde{M}_{k, m}}
            = \gamma_{k, m} \Tr{\widetilde{M}_{0, m}}, \label{eq:Mkm dev}
    \end{align}
    with $\Tr{\widetilde{M}_{k, m}}$-a.e.\ $\Tr{\gamma_{k, m}} = 1$. Note that the consistency equation \eqref{eq:cstc tilde} implies that $\Tr{\widetilde{M}_{k, m}}$ is independent of $k$. Hence,
    \begin{align*}
        \gamma_{k+1, m}^{(1, k)} \Tr{\widetilde{M}_{0, m}} = \gamma_{k, m} \Tr{\widetilde{M}_{0, m}},
    \end{align*}
    and $(\gamma_{k, m})_{k\in \mbN}$ is consistent $\Tr{\widetilde{M}_{0, m}}$-a.e. Then, for $\Tr{\widetilde{M}_{0, m}}$-a.e.\ $\lambda \in \sbra{0, 1}$, we can apply Theorem~\ref{th:QDFm} to $\prth{\gamma_{k, m}(\lambda)}_{k\in \mbN}$: there exist
    \begin{itemize}
        \item a probability measure $\mbP_{\lambda, m} \in \mcM\prth{S^m,\mbR_+}$,
        \item a function $\zeta_m(\lambda, \bullet) \in L^1\prth{S^m, \mcL^1_+(\mcH_0)}$ satisfying, $\mbP_{\lambda, m}$-a.e., $\Tr{\zeta_m(\lambda, \bullet)} = 1$,
    \end{itemize}
    such that
    \begin{align}
        \forall k\in \mbN,\ \gamma_{k, m}(\lambda) = \intr_{S^m} \zeta_m(\lambda, u)\otimes p_u^{\otimes k} \ d\mbP_{\lambda, m}(u). \label{eq:DF lambda}
    \end{align}
    Here $S^m$ stands for the $m$-dimensional sphere in $P_m \mcH_s$. From \eqref{eq:localization}, \eqref{eq:Mkm dev} and then \eqref{eq:DF lambda}, we infer that
    \begin{align}
        \mathds{1}_{\mcH_0} \otimes P_m^{\otimes k} \ \gamma^{(1, k)} \ \mathds{1}_{\mcH_0} \otimes P_m^{\otimes k} 
            =&\ \intr_0^1 \lambda^k dM_{k, m}(\lambda) 
            = \intr_0^1 \lambda^k d\widetilde{M}_{k, m}(\lambda) 
            = \intr_0^1 \lambda^k \gamma_{k, m}(\lambda) d\Tr{\widetilde{M}_{0, m}}(\lambda) \nonumber \\
            =&\ \intr_0^1 \lambda^k \prth{\ \intr_{S^m} \zeta_m(\lambda, u)\otimes p_u^{\otimes k}  \ d\mbP_{\lambda, m}(u)} d\Tr{\widetilde{M}_{0, m}}(\lambda). \label{eq:def measure}
    \end{align}

\subProof{Construction of a measure on $\sbra{0, 1}\times S^m$}

    We observe that, by testing \eqref{eq:def measure} against $\mathds{1}_{\mcH_0} \otimes \bigotimes_{i=1}^k A_i$ for $A_{1:k} \in \mcL\prth{\mbC^{m+1}}$ self-adjoint and $k \in \mbN$, we find, by density (Stone--Weierstrass theorem) of the algebra generated by functions of the form
    \begin{align*}
        \lambda, u \mapsto \lambda^k \prod_{i=1}^k \bk{u}{A_i u} \in C^0\prth{\sbra{0, 1}\times \quotient{S^m}{S^1}, \mbR},
    \end{align*}
    that we can define $\mbP_m \coloneq \mbP_{\bullet, m} \otimes \Tr{\widetilde{M}_{0, m}}$ as a measure on $\sbra{0, 1} \times S^m$, which is uniquely determined on measurable rectangles by
    \begin{align*}
        \forall A \in \mcB(\sbra{0, 1}),\ B \in \mcB(S^m),\ 
        \mbP_m(A \times B) \coloneq \intr_{\sbra{0, 1}} \mbP_{\lambda, m}(B) \, d\Tr{\widetilde{M}_{0, m}}(\lambda).
    \end{align*}
    This measure is finite, due to \eqref{eq:Mkm mass}:
    \begin{align}
        \mbP_m(\sbra{0, 1} \times S^m) 
            = \intr_{\sbra{0, 1}} d\Tr{\widetilde{M}_{0, m}}(\lambda) 
            = \Tr{\widetilde{M}_{0, m}(\sbra{0, 1})}
            \le 1. \label{eq:Pm mass}
    \end{align}
    We can extend $\mbP_m$ by $0$ to a measure on $\sbra{0, 1} \times S_{\mcH_s}$ since $S^m \subseteq P_m \mcH_s$. In particular,
    \begin{align}
        \mathds{1}_{\mcH_0} \otimes P_m^{\otimes k} \ \gamma^{(1, k)} \ \mathds{1}_{\mcH_0} \otimes P_m^{\otimes k} 
            = \iintr_{\sbra{0, 1} \times S_{\mcH_s}} \lambda^k \zeta_m(\lambda, u) \otimes p_u^{\otimes k} \, d\mbP_{m}(\lambda, u), \label{eq:finite m DF}
    \end{align}
    with $\mbP_m$-a.e., $\Tr{\zeta_m} = 1$. As a consequence of the above density argument and the $S^1$-invariance of Theorem~\ref{th:QDFm}, $\zeta_m \mbP_{m}$ is the unique $S^1$-invariant trace-class-valued measure satisfying \eqref{eq:finite m DF}.

\subProof{Taking $m\to\infty$}

    It follows from this uniqueness property that $(\zeta_m \mbP_m)_{m\in \mbN}$ is consistent in the following sense
    \begin{align*}
        \forall m \in \mbN,\ \zeta_m \mbP_m = {P_m}_* \prth{\zeta_{m+1} \mbP_{m+1}}
    \end{align*}
    and a bounded sequence in $\mcM\prth{\sbra{0, 1} \times S_{\mcH_s}, \mcL^1\prth{\mcH_0}}$. Here the Kolmogorov extension theorem holds. For this we refer to \cite[Chapters V.1 and V.2]{VectorMeasures}, key ingredients are the Radon--Nikodym property of the trace-class and $\mbP_m$-a.e., $\Tr{\zeta_m} = 1$. Therefore the Kolmogorov extension theorem constructs a trace-class valued measure $\zeta \mbP$, with
    \begin{itemize}
        \item $\mbP \in \mcM\prth{\sbra{0, 1} \times S_{\mcH_s},\mbR_+}$, satisfying $\mbP\prth{\sbra{0, 1} \times S_{\mcH_s}} \le 1$,
        \item $\zeta \in L^1\prth{\sbra{0, 1} \times S_{\mcH_s}, \mcL^1_+(\mcH_0)}$ satisfying, $\mbP$-a.e., $\Tr{\zeta} = 1$, whose existence is guaranteed by the Radon--Nikodym theorem,
    \end{itemize}
    such that 
    \begin{align*}
        \forall m \in \mbN,\ \zeta_m \mbP_m = {P_m}_* \prth{\zeta \mbP}.
    \end{align*}

    Let $K \in \mcK\prth{\mcH_0 \otimes \mcH_s^{\otimes_+ k}}$. Since
    \begin{align*}
        \varphi: (\lambda, u) \mapsto \lambda^k\prth{\mathds{1}_{\mcH_0} \otimes p_u^{\otimes k}\ K}^{(1, 0)} \in C^0\prth{\sbra{0, 1} \times S_{\mcH_s}, \mcK\prth{\mcH_0}},
    \end{align*}
    we obtain
    \begin{align*}
        \iintr_{\sbra{0, 1} \times S_{\mcH_s}} \lambda^k\Tr{\zeta_m(\lambda, u)\otimes p_u^{\otimes k} \ K}  d\mbP_{m}(\lambda, u)
            =&\ \iintr_{\sbra{0, 1} \times S_{\mcH_s}} \Tr{\zeta_m(\lambda, u) \varphi(\lambda, u)}  d\mbP_{m}(\lambda, u) \\
            \cvgto_{m\to\infty}&\ \iintr_{\sbra{0, 1} \times S_{\mcH_s}} \Tr{\zeta(\lambda, u) \varphi(\lambda, u)}  d\mbP(\lambda, u) \\
            =&\ \iintr_{\sbra{0, 1} \times S_{\mcH_s}} \lambda^k\Tr{\zeta(\lambda, u)\otimes p_u^{\otimes k} \ K}  d\mbP(\lambda, u),
    \end{align*}
    from which we deduce that
    \begin{align}
        \iintr_{\sbra{0, 1} \times S_{\mcH_s}} \lambda^k\zeta_m(\lambda, u)\otimes p_u^{\otimes k}  d\mbP_{m}(\lambda, u) 
            \wslim\displaylimits_{m\to\infty} \iintr_{\sbra{0, 1} \times S_{\mcH_s}} \lambda^k\zeta(\lambda, u)\otimes p_u^{\otimes k}  d\mbP(\lambda, u). \label{eq:finite m DF limit}
    \end{align}

\subProof{Estimating the tail of $\mathds{1}_{\mcH_0} \otimes P_m^{\otimes k} \ \gamma^{(1, k)} \ \mathds{1}_{\mcH_0} \otimes P_m^{\otimes k}$}

    By induction,
    \begin{align*}
        \mathds{1}_{\mathcal{H}_s}^{\otimes k} - P_m^{\otimes k} 
            = \sum_{n=1}^{k}
        P_m^{\otimes (n-1)} \otimes Q_m \otimes \mathds{1}_{\mathcal{H}_s}^{\otimes (k-n)}.
    \end{align*}
    Indeed, the case $k=1$ corresponds to the definition of $Q_m$, and assuming the above identity, we obtain
    \begin{align*}
        \mathds{1}_{\mathcal{H}_s}^{\otimes (k+1)} - P_m^{\otimes (k+1)} 
            &= P_m \otimes \prth{ \mathds{1}_{\mathcal{H}_s}^{\otimes k} - P_m^{\otimes k}} + Q_m \otimes \mathds{1}_{\mathcal{H}_s}^{\otimes k} \\
            &= \sum_{n=1}^{k}
        P_m^{\otimes n} \otimes Q_m \otimes \mathds{1}_{\mathcal{H}_s}^{\otimes (k-n)} + Q_m \otimes \mathds{1}_{\mathcal{H}_s}^{\otimes k}
            = \sum_{n=0}^{k}
        P_m^{\otimes n} \otimes Q_m \otimes \mathds{1}_{\mathcal{H}_s}^{\otimes (k-n)} \\
            &= \sum_{n=1}^{k+1}
        P_m^{\otimes (n-1)} \otimes Q_m \otimes \mathds{1}_{\mathcal{H}_s}^{\otimes (k+1-n)}.
    \end{align*}
    Consequently,
    \begin{align*}
        &\norm{\mathds{1}_{\mathcal{H}_0}\otimes (\mathds{1}_{\mathcal{H}_s}^{\otimes k} - P_m^{\otimes k})\ \gamma^{(1, k)}}_{\mathcal{L}^1}
            \le \sum_{n=1}^k \norm{\mathds{1}_{\mathcal{H}_0} \otimes P_m^{\otimes (n-1)} \otimes Q_m \otimes \mathds{1}_{\mathcal{H}_s}^{\otimes (k-n)} \ \gamma^{(1, k)}}_{\mathcal{L}^1} \\
            &\le \sum_{n=1}^k \Tr{\mathds{1}_{\mathcal{H}_0} \otimes \mathds{1}_{\mathcal{H}_s}^{\otimes (n-1)} \otimes Q_m \otimes \mathds{1}_{\mathcal{H}_s}^{\otimes (k-n)} \ \gamma_{1, k}}^{\frac{1}{2}} \Tr{\mathds{1}_{\mathcal{H}_0} \otimes P_m^{\otimes (n-1)} \otimes \mathds{1}_{\mathcal{H}_s}^{\otimes (k+1-n)} \ \gamma_{1, k}}^{\frac{1}{2}} \\
            &\le \sum_{n=1}^k \Tr{\mathds{1}_{\mathcal{H}_0} \otimes \mathds{1}_{\mathcal{H}_s}^{\otimes (n-1)} \otimes Q_m \otimes \mathds{1}_{\mathcal{H}_s}^{\otimes (k-n)} \ \gamma_{1, k}}^{\frac{1}{2}} 
            = k \Tr{Q_m \gamma^{(0, 1)}}^{\frac{1}{2}} \cvgto_{m\to\infty} 0.
    \end{align*}
    Developing the square, it follows from the Cauchy--Schwarz inequality that
    \begin{align*}
        \norm{\gamma^{(1, k)} - \mathds{1}_{\mcH_0} \otimes P_m^{\otimes k} \ \gamma^{(1, k)} \ \mathds{1}_{\mcH_0} \otimes P_m^{\otimes k}}_{\mcL^1}
            \le&\ \norm{\mathds{1}_{\mathcal{H}_0}\otimes (\mathds{1}_{\mathcal{H}_s}^{\otimes k} - P_m^{\otimes k})\ \gamma^{(1, k)} \ \mathds{1}_{\mathcal{H}_0}\otimes (\mathds{1}_{\mathcal{H}_s}^{\otimes k} - P_m^{\otimes k})}_{\mathcal{L}^1} \\
            &+ 2 \norm{\mathds{1}_{\mathcal{H}_0}\otimes (\mathds{1}_{\mathcal{H}_s}^{\otimes k} - P_m^{\otimes k})\ \gamma^{(1, k)} \ \mathds{1}_{\mathcal{H}_0}\otimes P_m^{\otimes k}}_{\mathcal{L}^1} \\
            \cvgto_{m\to\infty} 0.
    \end{align*}

\subProof{Conclusion}
    
    Combining the above with \eqref{eq:finite m DF} and \eqref{eq:finite m DF limit}, we obtain
    \begin{align*}
        \gamma^{(1, k)} = \iintr_{\sbra{0, 1} \times S_{\mcH_s}}  \lambda^k \zeta(\lambda, u)\otimes p_u^{\otimes k}  d\mbP(\lambda, u).
    \end{align*}
    Taking the trace yields
    \begin{align*}
        1 
            = \Tr{\gamma^{(1, k)}} 
            = \iintr_{\sbra{0, 1} \times S_{\mcH_s}}  \lambda^k d\mbP(\lambda, u).
    \end{align*}
    Since with \eqref{eq:Pm mass},
    \begin{align*}
        \mbP\prth{\sbra{0, 1} \times S_{\mcH_s}} 
            = \norm{\zeta \mbP}_{\mcM}
            \le \liminf_{m\to\infty} \norm{\zeta_m \mbP_m}_{\mcM}
            = \liminf_{m\to\infty} \mbP_m(\sbra{0, 1}\times S^m) 
            \le 1,
    \end{align*}
    we infer that $\mbP$ must be a probability measure supported on $\sett{1} \times S_{\mcH_s}$, and the pair
    \begin{align*}
        \mbP(\sett{1}\times \bullet), \zeta(1, \bullet)
    \end{align*}
    is suitable to conclude, as $\mbP(\sett{1}\times \bullet)$-a.e.\ $\Tr{\zeta(1, \bullet)} = 1$.

\end{proof}

\section*{Declarations}
\addcontentsline{toc}{section}{Declarations}

\noindent\textbf{Fundings and acknowledgments.} S.~Farhat and S.~Petrat acknowledge funding by the
Deutsche Forschungsgemeinschaft (DFG, German Research Foundation) - project number
505496137. D.~Périce and S.~Petrat acknowledge funding by the
Deutsche Forschungsgemeinschaft (DFG, German Research Foundation) - project number
512258249. This work was supported by the German Research Foundation (DFG) within the scientific network `A(E)MP - Appearance of the Effective Mass in Polaron Models´ (grant No. 569490025). 

We would like to thank the Institut Henri Poincaré (UAR 839 CNRS-Sorbonne Université) and the LabEx CARMIN (ANR-10-LABX-59-01) for hosting us during the thematic program on "Quantum many-body systems out-of-equilibrium", where some ideas of this paper were initially discussed. In particular, we would like to thank Thierry Giamarchi (one of the organizers of the thematic programme), for discussions on quantum lattice systems and DMFT.

\vspace{3mm}
\noindent\textbf{Data availability statement.} No datasets were generated or analysed during the current study.

\vspace{3mm}
\noindent\textbf{Conflicts of interests declaration.} The authors have no competing interests to declare that are relevant to the content of this article.

\addcontentsline{toc}{section}{References}
\bibliography{bibliography}

\end{document}